\documentclass[11pt]{article}
\usepackage[left=1in,top=1in,right=1in,bottom=1in,head=0in,nofoot]{geometry}

\usepackage{setspace,url,bm,amsmath} 
\usepackage{scalerel}

\usepackage{xcolor}
 \usepackage{multirow}
 \usepackage{arydshln}
 \usepackage{mathtools}
\usepackage{titlesec} 
\titlelabel{\thetitle.\quad} 
\titleformat*{\section}{\bf\Large\center}

\usepackage{authblk}
\usepackage{float}
\usepackage{graphicx} 
\usepackage{bbm}
\usepackage{latexsym}
\usepackage{caption}
\usepackage[margin=20pt]{subcaption}
\usepackage{enumerate}
\graphicspath{ {./Graphs/} }

\usepackage{tcolorbox}

\usepackage{amsthm}
\usepackage{amssymb}
\usepackage{amsmath}
\usepackage{color}

\usepackage{comment}
\theoremstyle{definition}
\newtheorem{assumption}{Assumption}
\newtheorem*{theorem*}{Theorem}
\newtheorem{theorem}{Theorem}
\newtheorem*{rmk*}{remark}

\newtheorem{lemma}{Lemma}
\newtheorem{example}{Example}

\newtheorem{condition}{Condition}

\newtheorem*{corollary*}{Corollary}

\usepackage{natbib} 
\bibpunct{(}{)}{;}{a}{}{,} 

\usepackage{etoolbox} 
\apptocmd{\sloppy}{\hbadness 10000\relax}{}{} 

\usepackage{multibib}
\newcites{sec}{References}

\usepackage{color}
\usepackage{listings}
\usepackage{hyperref}
\usepackage{booktabs}

\usepackage{lscape}

\RequirePackage[normalem]{ulem}

\def \P{\mathbb{P}}

\def \E {\mathbb{E}}
\def \cE{\mathcal{R}} 
\def \cO{\mathcal{E}}
 \def \cS{\mathcal{S}}

\newcommand{\indep}{\perp \!\!\! \perp}

\usepackage{wasysym}
\usepackage{algorithmic}
\usepackage{algorithm}

\allowdisplaybreaks

\begin{document}

\singlespacing

\title{\bf Improving Treatment Effect Estimation in Trials through Adaptive Borrowing of External Controls}

\author[1]{Qinwei Yang}
\author[2]{Jingyi Li}
\author[1]{Peng Wu\thanks{Corresponding author: pengwu@btbu.edu.cn}}
\author[3]{Shu Yang}
\affil[1]{Beijing Technology and Business University,  Beijing, China} 
\affil[2]{National University of Singapore, Singapore}  
\affil[3]{North Carolina State University, Carolinas, USA}


\date{}

\maketitle

\begin{abstract}
Randomized controlled trials (RCTs) often suffer from limited inferential efficiency in estimating treatment effects due to their small sample sizes. In recent years, incorporating external controls (ECs) has gained increasing attention as an effective way to augment small RCTs and thereby enhance estimation efficiency. However, ECs are not always comparable to RCTs, and direct borrowing without careful evaluation can introduce substantial bias and, paradoxically, undermine the accuracy of treatment effect estimation. 
In this paper, we propose a novel adaptive influence-based sample borrowing framework to improve average treatment effect (ATE) estimation in RCTs. The framework quantifies the ``comparability'' of each sample in ECs using influence functions and identifies the optimal subset of ECs that minimizes the mean squared error of the ATE estimator.
The proposed framework is assumption-lean regarding the distribution of ECs and is robust to outliers, making it broadly applicable across diverse settings. Moreover, we develop an outcome calibration method to improve the data utilization efficiency of ECs, further strengthening the adaptive influence-based sample-borrowing framework. 
We demonstrate the effectiveness of the proposed method using both simulated and real-world datasets. 
\end{abstract}



\newpage

\onehalfspacing


\section{Introduction}
Randomized controlled trials (RCTs) are regarded as the gold standard for estimating treatment effects, as randomization can effectively eliminate confounding between treatment and outcome~\citep{imbens2015causal, Hernan-Robins2020}. However, their inferential efficiency is typically limited by small sample sizes, due to high costs and lengthy recruitment periods~\citep{athey2017state, hu2023longterm,  Wu-etal-ShortLong}.
To address this challenge, there has been growing interest in augmenting small RCTs by leveraging external datasets containing only control arms, referred to as \emph{ external controls} (ECs), to improve treatment effect estimation in RCTs. 
For example, in digital marketing, platforms may combine small-scale A/B test data with historical single-arm user behavior logs to enhance inference~\citep{bojinov2023design}. In medical research, oncology trials often incorporate historical control data to strengthen causal conclusions~\citep{hobbs2011hierarchical}. Similarly, rare disease studies frequently augment small randomized trials with matched ECs from disease registries to improve statistical power~\citep{wang2019using}.

To fully exploit ECs, it is indispensable to invoke several assumptions on the relationship between RCTs and ECs~\citep{gao2024improving, wu2024comparative}.  
Toward this, most studies rely on the exchangeability assumption, positing that the distribution of potential outcomes under control remains invariant between RCTs and ECs when conditioned on baseline covariates~\citep{dahabreh2019generalizing, Dahabreh-etal2020, li2023improving, Wu-Mao2025, Colnet-etal2024}.   
However, the exchangeability assumption requires that individuals in both RCT controls and ECs follow the same pattern between covariates and outcome, which may be less plausible in real-world applications due to individual heterogeneity~\citep{gao2024improving}. 
In practice, ECs usually have significantly larger sample sizes than RCTs and often exhibit greater individual heterogeneity~\citep{li2023improving, fda2023considerations}. When some individuals in ECs show patterns that differ from those in RCTs, the exchangeability assumption is violated and leads to biased conclusions.

\emph{In this article, we aim to improve the average treatment effect (ATE) estimation in an RCT by adaptively borrowing ECs without relying on the exchangeability assumption. The key lies in (a) measuring the ``comparability'' of each sample in ECs for the ATE estimation, thereby distinguishing between comparable (or beneficial) and non-comparable (or harmful) samples in ECs, and (b) defining and identifying the optimal set of comparable samples in ECs.}  


\emph{For the first goal (a)}, we propose an adaptive influence-based sample borrowing approach. Specifically, we employ the influence function~\citep{Hampel, koh2017understanding} to quantify how each sample in ECs, when added to the RCTs, perturbs the outcome model fitted on RCT controls, thereby yielding influence scores that reflect the comparability of each sample in ECs.  
Intuitively, a sample with a small influence score has less impact on the outcome model in RCT controls and is therefore considered more comparable. Including such samples in the RCT can effectively increase the sample size of controls without introducing bias. In contrast, a large influence score suggests that including the sample would substantially affect the outcome model, introducing bias into the treatment effect estimation. 
Using ranked influence scores, we construct nested candidate subsets of ECs, where the $k$-th subset contains the top-$k$ ECs with the smallest influence scores. 


\emph{For the second goal (b)}, we first derive the semiparametric efficient estimator~\citep{tsiatis2006semiparametric} of the ATE that combines the RCT with an arbitrary candidate subset of ECs under the exchangeability assumption.  
Such an estimator minimizes the asymptotic variance among all regular estimators and is often considered optimal~\citep{newey1990semiparametric, van2000asymptotic} under the given assumptions. 
We establish the consistency and asymptotic normality of the proposed estimator. 
Then, recognizing that the exchangeability assumption may not hold for all candidate subsets of ECs, we analyze the asymptotic bias and variance of the proposed estimator under violations of exchangeability. 
Finally, we propose selecting the optimal subset of ECs that yields the ATE estimator with the minimum mean squared error (MSE). 

By integrating the two components developed for goals (a) and (b), we establish an adaptive influence-based sample borrowing framework to improve ATE estimation in RCTs. This framework consists of two steps: obtain the influence scores for each sample in ECs; identify the optimal subset of ECs that minimizes the MSE of the ATE estimator. The proposed framework remains effective when the exchangeability assumption is violated.

However, a careful examination of the borrowing behavior of the proposed influence-based framework reveals that it is effective only when several samples in ECs have small influence scores. When most samples in ECs have large influence scores---that is, when the majority of ECs exhibit outcome model patterns that differ substantially from those in RCT controls---the framework can borrow only a small number of ECs, resulting in an estimator similar to that based solely on RCT data and limiting data utilization efficiency. 
This is intuitive, as including such samples would significantly increase bias while providing little reduction in variance.

To mitigate this, \emph{we further propose a novel outcome calibration method to improve the data utilization efficiency of ECs.} We achieve this by adjusting for outcome differences between  RCT controls and ECs. 
We first propose a robust bias estimation approach that combines RCT controls and ECs, and then calibrate the outcomes in ECs by subtracting the estimated bias for each sample. This adjustment ensures that the conditional means of the potential outcomes under control are aligned between the RCT controls and ECs. After calibration, we then apply the adaptive influence-based sample borrowing framework to obtain the final estimator of ATE in RCT data.   
The main contributions are summarized as follows.
\begin{itemize}
\item We reveal the limitations of existing approaches for improving the estimation of ATEs by combining RCTs and ECs.

\item We propose an adaptive influence-based sample borrowing framework to improve ATE estimation in RCTs. 
This framework is established by novelly using influence scores to quantify the comparability of ECs and by borrowing the optimal subset of ECs based on MSE.


\item We propose a calibration method to enhance the data utilization efficiency of ECs, thereby further strengthening the adaptive influence-based sample borrowing framework.


\item We conduct extensive experiments on both simulated and real-world datasets, demonstrating that the proposed approach outperforms existing baseline methods.

\end{itemize}

\section{Related Work}  \label{related-work}

Since the work of \cite{pocock1976combination}, there has been increasing research on augmenting RCT data with ECs to improve the efficiency of treatment effect estimation. These methods are termed \textit{external control}, \textit{historical control}, or \textit{history borrowing} in the related literature. 
Under the exchangeability assumption (or its analog) for the potential outcome under control, several approaches~\citep{dahabreh2019generalizing, li2023improving, valancius2024causal} have been proposed to estimate treatment effects in RCTs using ECs. However, the exchangeability assumption may be less plausible in practice due to individual heterogeneity. As a result, combining RCT data and ECs directly will lead to bias. 


\textbf{Selective Borrowing}: To relax the exchangeability assumption, several methods have been proposed to adaptively select the degree of borrowing or adjust EC outcomes based on observed differences from RCT controls. These methods include matching and bias adjustment~\citep{stuart2008matching}, power priors~\citep{neuenschwander2009note}, bayesian hierarchical models such as meta-analytic predictive priors~\citep{schoenfeld2019design}, and commensurate priors~\citep{hobbs2011hierarchical}. 
Despite their appeal, previous simulation studies reveal that no single approach performs well in all scenarios when the exchangeability assumption is violated~\citep{shan2022simulation, gao2024improving}.  
Moreover, many of them rely on strong parametric model  assumptions, making them vulnerable to model misspecification and limiting their ability to 
accommodate the complex patterns typical of real-world applications~\citep{shan2022simulation,gao2024improving}.  
More recently, \cite{gao2024improving} proposed a state-of-the-art adaptive lasso–based method that selects a comparable subset of ECs through bias penalization. 

In this article, we extend this class of methodologies by proposing a novel adaptive influence-based sample borrowing framework, which quantifies the comparability of ECs through influence scores and identifies the optimal subset by minimizing the MSE. 
Extensive experiments show that the proposed method performs well across diverse scenarios. 

\textbf{Data Combination}: In addition to combining RCT data with ECs, there are other settings for data combination~\citep{Degtiar-Rose2023, kallus2024role, Colnet-etal2024, wu2024comparative}. For example, one can combine RCT (or experimental) data with external data that contain only covariates~\citep{Dahabreh-etal2019, Dahabreh-etal2020}. In such settings, the goal is typically to generalize the causal effect from the RCT to the external population, rather than to improve estimation of treatment effects within the RCT itself. Another common setting involves combining RCT data with confounded observational data that suffer from unobserved confounding~\citep{Chen-Cai-2021, hu2023longterm, Wu2023Transfer, Yang-etal2023}. In addition, several studies have combined multiple datasets for causal inference~\citep{Han-etal-2023, Wu-Mao2025}. When combining multiple datasets, it is important not only to consider which identifiability assumptions to adopt in order to improve estimation of treatment effects, but also to address how to preserve privacy. 
 
 Unlike these studies, which either use the entire external dataset or do not use external data at all, our method performs individual-level borrowing of external data.  
In addition, different from the above two classes of studies, we further propose a calibrated method to enhance data utilization efficiency. This method is model-agnostic and can be integrated with other approaches, granting it broad applicability.



The rest of the paper is organized as follows. Section \ref{sec3} introduces the notation and defines the problem of interest. Section \ref{motivation} provides the motivation for this work. Sections \ref{influence-function-framework} and \ref{sec6} present the adaptive influence-based sample borrowing framework, and Section \ref{sec-calibration} proposes the calibrated method. Section \ref{sec-illustrate} illustrates the sample borrowing behavior of the proposed methods through two simulated examples. Section \ref{sec-experiment} evaluates the numerical performance of the proposed methods via extensive experiments on both simulated and real-world datasets. Section \ref{sec-conclusion} concludes with a discussion.

\section{Preliminaries} \label{sec3}
\subsection{Notation and Setup}
Let \( X \in \mathcal{X} \subset \mathbb{R}^p \) denote the observed pre-treatment covariates and \( Y \in \mathcal{Y} \subset \mathbb{R} \) the outcome of interest. The binary treatment indicator is \( A \in \{0,1\} \), where \( A = 1 \) indicates treatment and \( A = 0 \) indicates control. Under the potential outcomes framework~\citep{rubin1974estimating, splawa1990application}, each individual has two potential outcomes, \( Y(0) \) and \( Y(1) \), corresponding to the outcome under control and treatment, respectively. We maintain the stable unit treatment value assumption~\citep{rubin1980randomization}, the  observed outcome satisfies $Y = (1 - A)Y(0) + A Y(1).$

Suppose we have access to RCT data and external data containing only control samples (ECs). The RCT data and ECs are denoted by 
 \begin{equation*}
  \begin{split}
      &\{X_i, A_i, Y_i, R_i = 1, i \in \mathcal{R} \},  \\
   &\{X_j, A_j = 0, Y_j, R_j = 0, j \in \mathcal{E} \},  
\end{split} 
\end{equation*}
 where \( R \) is a data source indicator: \( R = 1 \) corresponds to RCT data and \( R = 0 \) to the ECs. Let \( N_{\mathcal{R}} \) and \( N_{\mathcal{E}} \) denote the sample sizes of RCT data and ECs, respectively. 
  Let \( \mathbb{P}(\cdot \mid R=1) \) and \( \mathbb{P}(\cdot \mid R=0) \) represent the population distributions of RCT data and ECs, respectively, and let \( \mathbb{E} \) denote the expectation under \( \mathbb{P} \). 
The causal estimand of interest is the ATE in the RCT population,  defined as $\tau=\mathbb{E}[Y(1)-Y(0) \mid R=1]$.

\subsection{Estimation Based Only on RCT data}
The identification of $\tau$ from RCT data is guaranteed under the strong ignorability assumption, a standard condition in causal inference~\citep{rosenbaum1983central,imbens2004nonparametric}. 

\begin{assumption}[Strong ignorability for RCT data]\label{assump1} \quad  \par 
 (a) $A \indep \{Y(0),Y(1)\} \mid (X,R=1)$; \par (b) $0 < e_1(x) \triangleq \mathbb{P}(A=1 \mid X = x, R=1) < 1$ for all $x$ satisfying $\P(X=x\mid R=1)>0$, where $e_1(x)$ is the propensity score. 
\end{assumption}

Assumption \ref{assump1}(a) implies that $X$ captures the confounding between $A$ and $Y$. Assumption \ref{assump1}(b) ensures that individuals with $X = x$ have a positive probability of receiving each treatment option. Assumption \ref{assump1} is inherently satisfied in RCTs due to the randomization mechanism. Under Assumption \ref{assump1}, we can identify $\tau$ using only RCT data: 
$$\tau = \mathbb{E}[\mu_1(X)-\mu_0(X) \mid R=1 ],$$ 
where $\mu_a(X)=\mathbb{E}[Y\mid X,A=a,R=1]$ for $a=0,1,$ denote the outcome regression functions in RCT data.

When only RCT data are available, the semiparametric efficient estimator of $\tau$ under Assumption~\ref{assump1}  is given by
\begin{equation*}
    \hat \tau_{\textup{aipw}} = \frac{1}{N_{\cE}} \sum_{i\in \cE} \biggl\{\frac{ A_i(Y_i - \hat \mu_1(X_i))}{\hat e_1(X_i)} - \frac{(1-A_i)(Y_i - \hat \mu_0(X_i))}{1 - \hat e_1(X_i)} \\
    + (\hat \mu_1(X_i) - \hat \mu_0(X_i))\biggr\},
\end{equation*}

where $\hat \mu_0(x)$, $\hat \mu_1(x)$, and $\hat e_1(x)$ are estimates of $\mu_0(x)$, $\mu_1(x)$, and $e_1(x)$, respectively. 
This is the classical augmented inverse probability weighting (AIPW) or doubly robust estimator, which has been widely studied in prior work~\citep{Bang-Robins-2005, Kang-Schafer-2007, Tan2007, Wu-Han2024, Wu-Tan2024}.
The propensity score $e_1(x)$ may be known a priori in RCTs~\citep{gao2024improving, Qiu-etal2015}. In such cases, we set $\hat e_1(x) = e_1(x)$ and the proposed method in this article is directly applicable.
 
However, estimating $\tau$ using only RCT data often suffers from low efficiency due to small sample sizes, which result from high costs, long recruitment periods, and ethical or feasibility constraints. This paper aims to improve the estimation efficiency of $\tau$ by fully leveraging ECs.

\section{Motivation}
\label{motivation}

To leverage ECs, most studies rely on the exchangeability assumption below~\citep{dahabreh2019generalizing, Dahabreh-etal2020, wu2024comparative}. 

\begin{assumption}[Exchangeability] \label{mean-exchange}\quad  \par 
   (a) $R \indep Y(0) \mid X$ in $\P$; 
   \par(b) $\mathbb{P}(R = 1 \mid X) >0$ for all $X$ such that $\P(X, R=1)> 0$. 
\end{assumption}
Assumption~\ref{mean-exchange}(a) ensures the transportability of the conditional mean of $Y(0)$ is identical between RCT data and ECs, i.e., 
 $\mathbb{E}[Y(0)\mid X=x,R=0]=\mathbb{E}[Y(0)\mid X=x,R=1]$. 
This establishes a connection between RCT data and ECs.
Assumption~\ref{mean-exchange}(b) indicates that all samples in RCT data have a positive probability of
belonging to EC data.  
Under Assumption~\ref{mean-exchange}, we can use the \emph{entire} ECs to improve the treatment effect estimation in RCTs~\citep{li2023improving, gao2024improving, wu2024comparative}.

However, the exchangeability assumption is usually implausible in practice due to individual heterogeneity. 
Therefore, directly integrating ECs without proper scrutiny can lead to biased estimates. To address this,~\cite{gao2024improving} proposed a state-of-the-art,  adaptive lasso-based approach that identifies the optimal subset of ECs---rather than using the entire set---for integration, thereby improving treatment effect estimation in RCTs under weaker conditions. 
For clarity, we provide a brief overview of this approach and analyze its limitations, which motivate our work.  

Specifically, \cite{gao2024improving} introduced a vector of bias parameter $\bm{b}_0=(b_{1,0},\dots,b_{N_{\mathcal{E},0}})$ for all $j\in \mathcal{E}$, where 
$$b_{j,0}=\mu_{0,\mathcal{E}}(X_j)-\mu_{0}(X_j),$$
with $\mu_{0,\mathcal{E}}(X_j)=\mathbb{E}(Y_j\mid X_j, R_j=0)$. 
The bias $\bm{b}_0$ quantifies the difference in conditional mean outcomes between ECs and RCT controls. \cite{gao2024improving} considered samples in ECs with zero bias as ``comparable" and aimed to identify these samples from the pool of ECs.   
Let $\hat{b}_{j, 0} = \hat{\mu}_{0,\mathcal{E}}(X_j)-
\hat{\mu}_0(X_j)$ be a consistent estimator of $b_{j,0}$,  and $\hat{\bm{b}} = (\hat b_{1,0}, \cdots, \hat b_{N_{\mathcal{E},0}})$ an initial estimator of  $\bm{b}_0$. A sparse estimator $\tilde{\bm{b}}$ is then obtained by minimizing the adaptive lasso loss:  
    \begin{equation}\label{eq_alb}
\tilde{\bm{b}} = \arg\min_{\bm{b}} \left\{ (\widehat{\bm{b}} - \bm{b})^{\mathrm{T}} \widehat{\Sigma}_{\bm{b}}^{-1} (\widehat{\bm{b}} - \bm{b}) + \lambda \sum_{j \in \mathcal{E}} \frac{|b_{j,0}|}{|\hat{b}_{j,0}|^\nu} \right\}, 
\end{equation}
where $\widehat{\Sigma}_{\bm{b}}$ is the variance estimate of $\hat{\bm{b}}$ and $(\lambda,\nu)$ are two tuning parameters. Finally, the samples in ECs with an estimated zero bias are borrowed. 

\begin{figure*}[t!]
    \centering
    \subfloat[]{
    \begin{minipage}[t]{0.4\linewidth}
    \centering
    \includegraphics[width=1.0\textwidth]{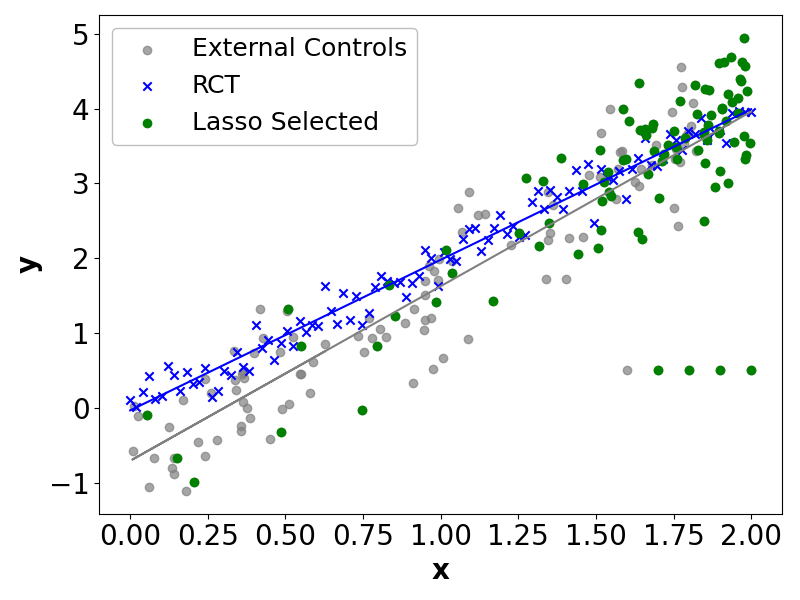}
    \end{minipage}%
    
    \begin{minipage}[t]{0.4\linewidth}
    \centering
    \includegraphics[width=1.0\textwidth]{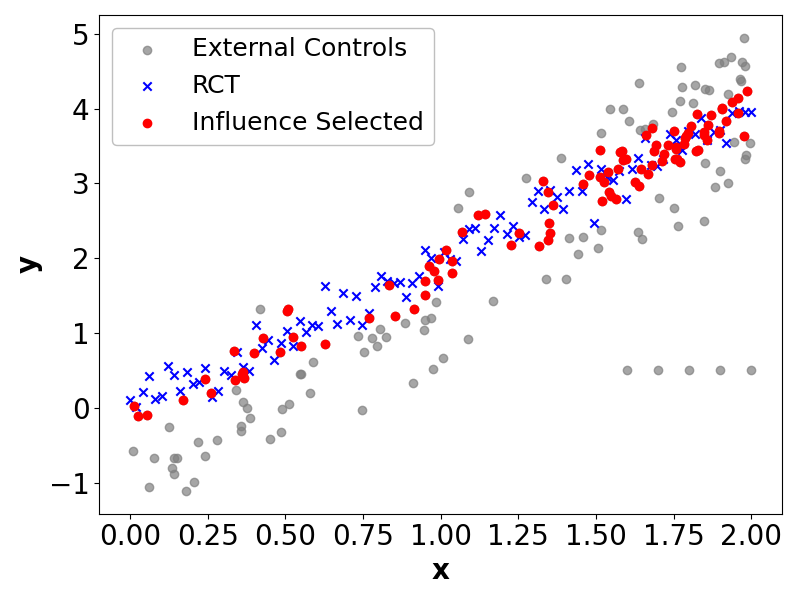}
    \end{minipage}%
    }

    \subfloat[]{
    \begin{minipage}[t]{0.4\linewidth}
    \centering
    \includegraphics[width=1.0\textwidth]{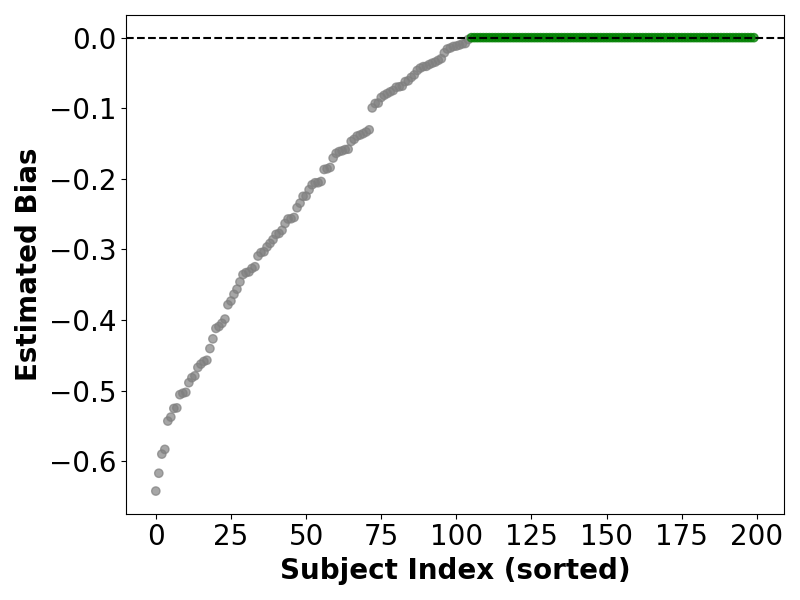}
    \end{minipage}%
    \begin{minipage}[t]{0.4\linewidth}
    \centering
    \includegraphics[width=1.0\textwidth]{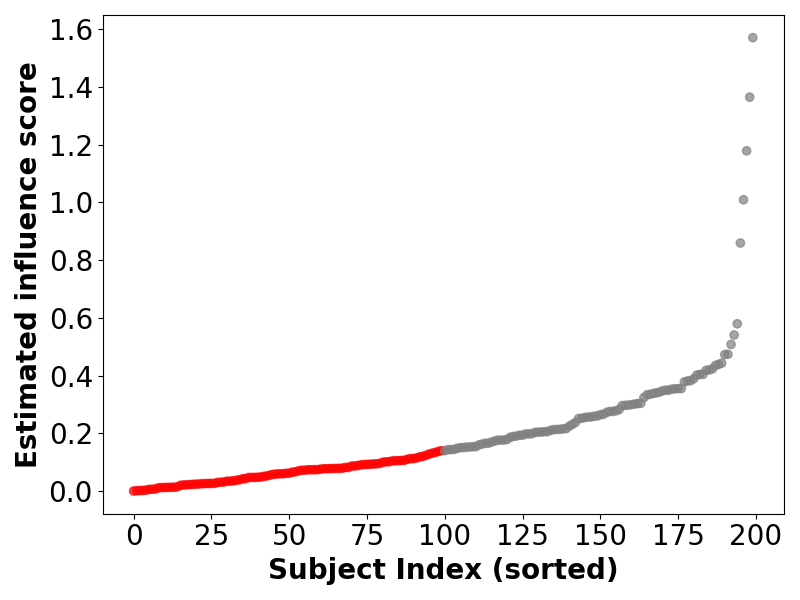}
    \end{minipage}%
    }%
    \centering
    \caption{Comparison of different approaches in a simulated example, see Example 1 in Section \ref{sec-illustrate} for details. (a) \textit{Sample Borrowing Comparison}: The left panel shows the adaptive lasso-based borrowing approach~\citep{gao2024improving}, while the right panel shows the influence-based borrowing approach proposed in Section \ref{sec6}.  
    (b) \textit{Bias and Influence Score Estimation}: The x-axis is the subject index, sorted by the estimated bias (left) or influence score (right). The left panel displays bias estimation for samples borrowed by the lasso-based approach (green points), while the right panel shows influence score estimation for samples borrowed by the influence-based approach (red points).}

    

    \label{fig.1}
    \setlength{\abovecaptionskip}{0.cm}
    \vspace{-0pt}
\end{figure*}

To better understand the borrowing behavior of the adaptive lasso-based approach, we conduct a simulation using a simple simulated dataset (see Example 1 in Section \ref{sec-illustrate} for details). 
In Fig.~\ref{fig.1}(a), the `x-shaped' and `dot-shaped' markers represent RCT controls and ECs, and the green points denote the ECs borrowed by the adaptive lasso-based approach. From it, we observe several limitations of this method: 
\begin{itemize}
\item Suboptimal comparability:  From Fig.~\ref{fig.1}(a), the samples (marked in green) borrowed by the adaptive lasso-based approach may exhibit patterns that differ substantially from those in RCT controls, resulting in suboptimal comparability. Intuitively, the adaptive lasso-based approach borrows samples with small values of $b(x) = \mu_{0,\mathcal{E}}(x) - \mu_0(x)$, where $b(x)$ represents an average discrepancy that overlooks individual-level heterogeneity in the neighborhood of $x$. Consequently, the approach may include samples with high variability (even picking out outliers), simply because their mean outcomes are close to those in RCT controls. 

\item Sensitivity to outliers: 
The final estimator $\tilde{\bm{b}}$ heavily depends on the estimate of $\hat{\mu}_{0,\mathcal{E}}(x)$, which is learned using \emph{all} ECs. When several outliers are present in ECs, the estimate of $\hat{\mu}_{0,\mathcal{E}}(x)$ may become biased, introducing estimation errors for individuals and exacerbating non-comparability issues. 
\end{itemize}

These limitations substantially hinder the effective use of ECs and restrict their potential to enhance the efficiency of treatment effect estimation in RCTs. Therefore, developing an improved data-borrowing methodology is essential.

\section{Adaptive Sample Borrowing Framework}
\label{influence-function-framework}

In this section, we propose an adaptive influence-based sample borrowing framework that more effectively identifies comparable ECs and enhances the efficiency of treatment effect estimation in RCTs.

\subsection{Quantifying Comparability of ECs via Influence Score}
The key lies in quantifying the comparability of each EC sample to RCT controls without modeling $\mu_{0,\cO}(x)$. We formalize this by posing a counterfactual question: how would the parameters in $\mu_0(x)$ change if an EC sample were added? If the model parameters change minimally, the sample can be considered comparable and suitable for borrowing.

Denote $\mathcal{C} = \{(X_i, A_i = 0, Y_i, R_i = 1), i \in \mathcal{R}\}$ as the set of RCT controls, and let $N_{\mathcal{C}}$ denote its sample size. 
Given RCT controls $\{Z_i = (X_i, Y_i) : Z_i \in \mathcal{C}\}$, we model $\mu_{0}(x)$ as $\mu_0(x; \theta)$, where $\theta$ is the parameter and is 
 estimated by the empirical risk minimizer:
\begin{equation}\label{erm} 
    \hat{\theta}\overset{\text{def}}{=} \arg\min_{\theta\in\Theta} \sum_{Z_i\in\mathcal{C}} L(Z_i;\theta),
\end{equation}
where $L(Z_i;\theta)=(Y_i- \mu_{0}(X_i;\theta))^2$ is a loss function that is twice-differentiable and convex in $\theta$. 
For an added EC sample $z=(x,y)$, we denote the modified parameter as $\hat{\theta}_{+z}$, which is obtained by retraining the model after adding $z$ as 
$$
    \hat{\theta}_{+z}\overset{\textup{def}}{=}\arg\min_{\theta\in\Theta} \sum_{Z_i\in\mathcal{C}\cup z} L(Z_i;\theta).
$$ 
We can then measure the influence of the EC sample $z$ on the loss over RCT controls using
\[ 
  \sum_{Z_i\in\mathcal{C}}|L(Z_i,\hat{\theta}_{+z})-L(Z_i,\hat{\theta})|,
  \]
where a larger value indicates that the sample $z$ has a greater impact on $\hat{\mu}_0(x) = \mu_0(x; \hat{\theta})$. 

However, retraining the model for each added sample $z$ is prohibitively slow and computationally expensive. Instead, following \citep{cook1980characterizations}, we approximate how the model parameters change when $z$ is added to RCT controls. The key idea is to assign an infinitesimal weight $\epsilon$ to $z$, yielding new model parameter: 
$$
\hat{\theta}_{\epsilon,z} \overset{\mathrm{def}}{=}  \arg\min_{\theta \in \Theta} \frac{1}{ N_{\mathcal{C}}} \sum_{Z_i\in\mathcal{C}} L(Z_i,\theta) + \epsilon L(z,\theta).$$ 
 A classic result~\citep{koh2017understanding,cook1980characterizations} shows that the influence of upweighting $z$ on the parameter $\hat{\theta}$ is given by
\[ 
\left. \frac{d \hat{\theta}_{\epsilon,z}}{d \epsilon} \right|_{\epsilon = 0} = - H_{\hat{\theta}}^{-1} \, \nabla_{\theta} L(z, \hat{\theta}),
\] 
where $H_{\hat{\theta}}\overset{\mathrm{def}}{=} N_{\mathcal{C}}^{-1} \sum_{Z_i\in\mathcal{C}} \nabla_{\hat{\theta}}^{2} L(Z_{i}, \hat{\theta})$ is the Hessian matrix. 
Then, the influence of adding the EC sample $z$ on the loss of a single sample $Z_i\in\mathcal{C}$ can be expressed in a closed-form expression by applying the chain rule,
\begin{equation*}
    \begin{aligned}
\left. \frac{dL(Z_i, \hat{\theta}_{\epsilon, z})}{d\epsilon} \right|_{\epsilon = 0} = -\nabla_\theta L(Z_i, \hat{\theta})^\top H_{\hat{\theta}}^{-1} \nabla_\theta L(z,\hat{\theta}).
\end{aligned}
\end{equation*}
which is ``the first-order derivatives of loss at $Z_i$'' multiplied by ``the parameter change when we add the EC sample $z$''. 
Furthermore, we can quantify the total influence of $z$ on the entire RCT controls as the influence score, which is given by 
\begin{equation}\label{eq7}
    \begin{aligned}
        \mathcal{IF}(z)
        \overset{\text{def}}{=}\sum_{Z_i\in\mathcal{C}}\left|\nabla_\theta L(Z_i,\hat{\theta})^\top H_{\hat{\theta}}^{-1} \nabla_\theta L(z, \hat{\theta})\right|. 
    \end{aligned}
\end{equation}

For all ECs $\{Z_j=(X_j,Y_j),j\in\mathcal{E}\}$, 
we calculate the influence scores $\{\mathcal{IF}(Z_j) ,j\in\mathcal{E}\}$ according to Eq.~\eqref{eq7}. These scores quantify the comparability of ECs.

\subsection{Adaptive Influence-Based Sample Borrowing Framework}



The proposed adaptive sample borrowing framework consists of two main steps: 

 (1) Based on the ranking of influence scores, we construct a sequence of nested subsets of ECs, where each subset contains the top-$k$ most comparable ECs, that is, those with the smallest influence scores. 
 
  (2) Find the optimal $k$ that minimizes the mean squared error (MSE) of the estimator based on RCT data and the borrowed top-$k$ ECs. This procedure is detailed in Section~\ref{sec6}.  

We summarize the proposed framework in Algorithm \ref{algo1}. 

\begin{algorithm}[t]
\caption{Adaptive Influence-Based Sample Borrowing Framework (AIB)}
\label{algo1}
\begin{algorithmic}[1]
\STATE \textbf{Input:} The RCT data $\{X_i, A_i, Y_i, R_i = 1, i \in \mathcal{R} \}$, and ECs: $(X_j,A_j=0,Y_j,R_j=0,j\in\mathcal{E})$.

\STATE \emph{\bf Step 1:} Fit the outcome model $\mu_0(X)$ and obtain the model parameters $\hat{\theta}$ by the empirical risk minimizer~\eqref{erm}


\STATE \emph{\bf Step 2}: Calculate the influence score $\mathcal{IF}(z)$  for each EC sample $z$ sing Eq.~\eqref{eq7}, then rank the scores and constructing a sequence of nested subsets of ECs, $\mathcal{\bf S} = \{ \cS_k: k = 1, ..., N_{\mathcal{E}}\},$ where $\mathcal{S}_k$ contains the top-$k$ samples with the smallest influence scores.

\STATE 
{\bf Step 3:} Estimate the MSE of $\hat{\tau}$ using the RCT data and borrowing set $\mathcal{S}_k$, as described in Section \ref{section.6.2}, then find the  subset of ECs, $\mathcal{S}_k^*$, that minimizes the MSE.

\STATE \textbf{Return:} The comparable samples set
 $\mathcal{S}_k^*$ and the final estimator of $\tau$. 
\end{algorithmic}
\end{algorithm}

\emph{Unlike the adaptive lasso-based approach, the influence-based approach does not rely on modeling $\mu_{0, \mathcal{E}}(x)$. Moreover, the influence score is defined at the individual level, and each sample's score is 
unaffected by other points in ECs, making the approach robust to outliers.}  
As illustrated in the right panel of Fig.~\ref{fig.1}(a), the proposed method effectively borrows samples (in red) that are close to RCT controls while maintaining robustness to outliers in ECs.

We remark that, although the influence score is not a novel concept, to the best of our knowledge, we are the first to use it to quantify the comparability of ECs. 
In addition, modeling $\mu_0(x)$ with $\mu_0(x; \theta)$ may appear to require a parametric model specification for $\mu_0(x)$. However, this is not a major limitation: the influence function method has been successfully applied to various neural network architectures~\citep{koh2017understanding}.  


\section{Improved Estimation via Finding the Optimal Subset of ECs}  \label{sec6}

In this section, we provide a detailed description of Step 3 in the adaptive sample borrowing framework presented in Algorithm \ref{algo1}.  
Let $\mathcal{S} \subseteq \mathcal{E}$ denote the subset of ECs borrowed by influence scores, and let $N_{\mathcal{S}}$ be its sample size.  
For ease of presentation, we denote $\P_{\mathcal{S}}$ as the combined population, where $\P_{\cS}(\cdot\mid R=1)$ and $\P_{\cS}(\cdot \mid R=0)$ represent the distributions of RCT data and the borrowed ECs, respectively. 
The expectation operator under $\P_{\cS}$ is denoted by $\mathbb{E}_{\mathcal{S}}$. For clarity, Table \ref{tab.2} summarizes the nuisance parameters used in constructing the estimator of $\tau$ and in the corresponding theoretical analysis. 

\begin{table}[H] 
\centering
\caption{Nuisance parameters}
\setlength{\tabcolsep}{2pt}
\resizebox{1\linewidth}{!}
{\begin{tabular}{ll}
\hline 
$e_1(X)=\P(A=1 \mid X,R=1) = \P_{\cS}(A=1 \mid X,R=1)$   &  propensity score  in $\cE$    \\ 
$\mu_a(X)=\E(Y \mid X, A=a, R=1) = \E_{\cS}(Y \mid X, A=a, R=1)$  &  outcome regression function in $\cE$  \\
   $\pi(X) = \P_{\cS}(R=1\mid X)$    & sampling score in  $\cE \cup \cS$  \\ 
$e_{\cS}(X) =  e_1(X) \pi(X)$ &  propensity score  in $\cE \cup \cS$    \\
$m_a(X) = \E_{\cS} [Y \mid X, A=a]$ & outcome regression function in $\cE \cup \cS$ \\
\hline 
\end{tabular}}
\\  
By definition, $\mu_1(X) = m_1(X)$, whereas $\mu_0(X)$ may not equal $m_0(X)$. 
\label{tab.2}
\end{table}

In Subsection \ref{sec6-1}, we construct the estimator of $\tau$ by combining RCT data $\mathcal{R}$ with the borrowed ECs $\mathcal{S}$, under a weaker version of Assumption \ref{mean-exchange} (Assumption \ref{assump2} below). In Subsection \ref{section.6.2}, we analyze the bias and variance of the estimator when Assumption \ref{assump2} is violated, and then aim to find the optimal subset of ECs that minimizes the MSE.

\subsection{Estimation Based on Exchangeability for Borrowed ECs} \label{sec6-1}

To fully utilize the borrowed data, we aim to derive the semiparametric efficient estimator of $\tau$, which is regarded as optimal since it achieves the semiparametric efficiency bound---yielding the smallest asymptotic variance under standard regularity conditions~\citep{tsiatis2006semiparametric,newey1990semiparametric}. To construct the semiparametric efficient estimator, we typically first derive the efficient influence function\footnote{The concept of the efficient influence function is distinct from that of the influence score introduced in Section \ref{influence-function-framework}.} and the semiparametric efficiency bound.


\begin{assumption}[Exchangeability for borrowed ECs] \label{assump2} \quad  \par
    (a) $R \indep Y(0) \mid X$ in $\P_{\cS}$;  \par (b) $\pi(X)=\mathbb{P}_{\mathcal{S}}(R = 1 \mid X) >0$ for all $X$ such that $\mathbb{P}_{\mathcal{S}}(X,R=1)>0$.
\end{assumption}

Assumption~\ref{assump2} is analogous to, but weaker than, Assumption~\ref{mean-exchange}, differing only by replacing  $\P$ with $\P_{\cS}$. It means that the outcome means for $Y(0)$ given covariates are the same in $\mathcal{R}$ and  $\cS$, that is, $\E_{\cS} [ Y(0)\mid X, R=1 ] =\E_{\cS} [ Y(0)\mid X, R=0]$. 
Assumption~\ref{assump2} is plausible when $\mathcal{S}$ is appropriately borrowed by influence scores. For example, in the right panel of Fig.~\ref{fig.1}(a), the RCT controls (in blue) and the borrowed ECs (in red) are close, and follow a similar distribution, supporting the plausibility of Assumption~\ref{assump2}.  



\begin{lemma}\label{lem1} Under Assumptions \ref{assump1} and \ref{assump2}, the efficient influence function of $\tau$ is 
\begin{equation*}
    \phi = \frac{\pi(X)}{q} \bigg \{  \frac{ RA(Y - m_1(X))}{ e_{\mathcal{S}}(X)  } - \frac{(1-A)(Y - m_0(X))   }{1 - e_{\mathcal{S}}(X)} \bigg \} +  \frac{R}{q}\{ m_1(X) -  m_0(X) - \tau \},
\end{equation*}

where $e_{\cS}(X)$,  $\pi(X)$, 
$m_1(X)$, and $m_0(X)$ are defined in Table \ref{tab.2} and $q=\mathbb{P}_{\mathcal{S}}(R=1)$.
The associated semiparametric efficiency bound for $\tau$ is 
   $\textup{Var}(\phi)$. 
\end{lemma}


Lemma \ref{lem1} presents the efficient influence function of $\tau$ under Assumptions \ref{assump1} and \ref{assump2}. Based on it, we can construct the estimator of $\tau$ as follows 
  \begin{align*}   
   \hat\tau_{\mathcal{S}} ={}&  \frac{1}{ N_{\mathcal{R}} + N_{\mathcal{S}} } \sum_{i\in \mathcal{R}\cup\mathcal{S}}   \Biggl [ \frac{\hat \pi(X_i)}{q} \Big \{  \frac{ R_iA_i(Y_i - \hat m_1(X_i))}{ \hat e_{\mathcal{S}}(X_i)  }    -   \frac{(1-A_i)(Y_i - \hat m_0(X_i)) }{1 - \hat e_{\mathcal{S}}(X_i)}\Big \} \\
   {}& +  \frac{R_i}{q}\{ \hat m_1(X_i) - \hat m_0(X_i)\}   \Biggr ]. 
  \end{align*}
where $\hat \pi(x), \hat e_1(x), \hat m_a(x)$ are 
 estimates of $\pi(x), e_1(x), m_a(x)$ for $a= 0,1$, and $\hat e_{\cS}(x) = \hat e_1(x) \hat \pi(x)$.   
 For simplicity, let $n =  N_{\cE} + N_{\cS}$.
 



\begin{condition} \label{cond1}  For a function $f(X)$ with a generic random variable $X$, define its $L_2$-norm as $||f(X)||_2 = \{\int f^2(x) d \P(x)\}^{1/2}$, suppose that    
$|| \hat e_{1}(X) - e_{1}(X) ||_2 \cdot || \hat{m}_a(X) - m_a(X) ||_2 = o_\mathbb{P}(n^{-1/2})$ and $|| \hat{\pi}(X) - \pi(X) ||_2 \cdot || \hat{m}_a(X) - m_a(X) ||_2 = o_\mathbb{P}(n^{-1/2})$ for  $ a \in \{0, \ 1 \}$, and additional regularity conditions in Assumption A.1 of Appendix.
\end{condition}

\begin{lemma}  \label{lem2}
      Under Assumptions \ref{assump1} and \ref{assump2},
      if Condition \ref{cond1} holds,    
  the estimator $\hat \tau_{\cS}$ satisfies
    $$
        \sqrt{n}(\hat \tau_{\cS} - \tau) \xrightarrow{d} \mathcal{N}(0, \sigma^2),
    $$
    where $\sigma^2 = \textup{Var}(\phi)$ is the semiparametric efficiency bound of $\tau$, and $\xrightarrow{d}$ means convergence in distribution.
 \end{lemma}

Lemma \ref{lem2} establishes the consistency and asymptotic normality of the estimator $\hat \tau_{\cS}$. It further shows that $\hat{\tau}_{\mathcal{S}}$ is semiparametrically efficient under regularity conditions, provided that the nuisance parameters are estimated at a rate faster than $n^{-1/4}$. These conditions are standard and widely adopted in machine-learning-assisted causal inference~\citep{chernozhukov2018double, Semenova-Chernozhukov, Kennedy-2023, Wu-etal-2025-Safe}.

\begin{lemma} \label{lem3}
Under the conditions in Lemma \ref{lem2}, we have 
$\textup{asy.var}( \sqrt{n}\hat \tau_{\mathcal{S}}) \leq \textup{asy.var}( \sqrt{n}\hat \tau_{\textup{aipw}})$, where $\textup{asy.var}$ denotes the asymptotic variance.  
The efficiency gain $\textup{asy.var}( \sqrt{n}\hat \tau_{\textup{aipw}}) - \textup{asy.var}(\sqrt{n}\hat \tau_{\mathcal{S}})$  is 
\begin{align*} 
   \E_{\cS} \left [  \frac{ \pi(X) }{q^2}   \frac{\textup{Var}(Y(0)|X)}{ 1 - e_1(X) }  \frac{ \P_{\cS}(A=0, R=0|X) }{  \P_{\cS}(A=0|X)}\right ],
\end{align*}
where $\P_{\cS}(A=0, R=0|X) = \P_{\cS}(R=0|X) \P_{\cS}(A=0|R=0,X) = 1-\pi(X)$, $\P_{\cS}(A=0|X) = 1 - e_{\cS}(X)$. 
\end{lemma}

Lemma \ref{lem3} shows that if the borrowed ECs satisfy Assumption \ref{assump2}, then combining them with RCT data can improve the estimation efficiency of $\tau$ compared with using RCT data alone. 

We further provide a complementary remark regarding Assumption \ref{assump2}. A slightly weaker version of Assumption \ref{assump2} is
$\E_{\cS}[Y(0)\mid X, R=1] = \E_{\cS}[Y(0)\mid X, R=0]$.
All conclusions in Lemmas \ref{lem1}–\ref{lem3} are similar under this weaker assumption, so we don't distinguish between the two in detail for ease of presentation.  
\subsection{Borrowing of ECs}\label{section.6.2}

In Section \ref{sec6-1}, given a borrowed EC subset $\mathcal{S}$, we construct the estimator $\hat{\tau}_{\mathcal{S}}$ by combining data from $\mathcal{R}$ and $\mathcal{S}$. However, the validity of $\hat{\tau}_{\mathcal{S}}$ relies on Assumption \ref{assump2}, which may not hold for an arbitrary choice of $\mathcal{S}$. 
When Assumption \ref{assump2} is violated, $\hat{\tau}_{\mathcal{S}}$ is no longer a consistent estimator of $\tau$ and becomes biased.
Intuitively, borrowing ECs involves a bias–variance trade-off: increasing the sample size of $\mathcal{S}$ reduces variance but may introduce non-comparable samples and thus induce bias. The proposed sample borrowing method is based on the bias–variance analysis of $\hat{\tau}_{\mathcal{S}}$.



\begin{theorem}[Bias-variance analysis] \label{thm1} Under Assumption \ref{assump1} only, 
      if Condition \ref{cond1} holds, 
    then $\hat \tau_{\cS}$ satisfies
    $$
        \sqrt{n}\{ \hat \tau_{\cS} - \tau - \textup{bias}(\hat \tau_{\cS}) \} \xrightarrow{d} \mathcal{N}(0, \sigma^2),
    $$
    where $\sigma^2 = \textup{Var}(\phi)$, $\phi$ is defined in Lemma \ref{lem1}, and 
    $$\textup{bias}(\hat \tau_{\cS}) = \E_{\cS}\ [ \frac{R}{q} \{ \mu_0(X) - m_0(X)  \} ].$$ 
\end{theorem}

Theorem \ref{thm1} extends Lemma \ref{lem2} by allowing Assumption \ref{assump2} to be violated. When Assumption \ref{assump2} holds, the bias term vanishes, and Theorem \ref{thm1} reduces to Lemma \ref{lem2}, with the asymptotic variance $\sigma^2$ attaining the semiparametric efficiency bound. In contrast, when Assumption \ref{assump2} is violated, bias arises and is proportional to the expectation of $\mu_0(X) - m_0(X)$, and the asymptotic variance $\sigma^2$ retains the same form.

Based on Theorem \ref{thm1}, we propose borrowing the optimal set $\cS$ that minimizes the MSE of $\hat{\tau}_{\cS}$ from the candidate sets. Specifically, the MSE of $\hat \tau_{\cS}$ is $\textup{MSE}(\hat \tau_{\cS}) = \{\textup{bias}(\hat \tau_{\cS})\}^2 + \textup{var}(\hat \tau_{\cS})$, $\textup{var}(\hat \tau_{\cS}) = \sigma^2 /n$.  
We estimate the bias as 
   $\widehat{\textup{bias}}(\hat \tau_\cS) =  \hat{\tau}_{\mathcal{S}}-\hat \tau_{\textup{aipw}}$, 
and estimate the variance as
$\widehat{\textup{var}}(\hat \tau_\cS)  =  \hat \sigma^2/ n,$ 
    where $\hat \sigma^2$ is the sample variance of
	\begin{align*} 
	      & \biggl \{ \frac{\hat \pi(X_i)}{q} \left ( \frac{ R_iA_i(Y_i - \hat m_1(X_i))}{ \hat e_{\mathcal{S}}(X_i)  }  -   \frac{(1-A_i)(Y_i - \hat m_0(X_i)) }{1 - \hat e_{\mathcal{S}}(X_i)}\right )  
	     +  \frac{R_i}{q}\{ \hat m_1(X_i) - \hat m_0(X_i)\}, i  \in \cE \cup \cS  \biggr \}. 
	     \end{align*}
Then the estimated MSE is $\widehat{\text{MSE}}(\hat{\tau}_{\cS}) =  \{\widehat{\textup{bias}}(\hat \tau_{\cS})\}^2 + \widehat{\textup{var}}(\hat \tau_{\cS})$.

Our goal is to find the optimal subset of ECs, defined as  
 $$\cS^*= \arg \min_{\mathcal{S}_k  \in {\bf S} } \text{MSE}(\hat{\tau}_{\cS_k}),$$
where $\mathcal{\bf S} = \{ \cS_k: k = 1, ..., N_{\mathcal{E}}\}$, $\cS_k$ denotes a subset of ECs corresponding to the top-$k$ smallest influence scores. 
The $\cS^*$ may contain more than one subset of ECs; that is, there may exist multiple minima in terms of MSE.    
In real-world applications, we estimate $\cS^*$ with  
	\[\hat \cS  =  \arg \min_{\mathcal{S}_k  \in {\bf S} } \widehat{\text{MSE}}(\hat{\tau}_{\cS_k}).     \]
 

It is noteworthy that the proposed borrowing strategy for estimating $\cS^*$ does not rely on Assumption~\ref{mean-exchange} or~\ref{assump2}, nor does impose any strict restrictions on the distribution of ECs. Consequently, our method is widely applicable across a variety of settings.  

RCT data typically have relatively small sample sizes, as their collection is time-consuming and costly, and further constrained by the inclusion criteria set by the experimenter~\citep{Kallus-Puli2018}. In contrast, ECs are observational, more readily available, and often drawn from larger and more diverse sources, exhibiting greater individual heterogeneity~\citep{li2023improving, fda2023considerations, yang2024learning, Colnet-etal2024}. 
As a result, it is inevitable that there exist some individuals (or outliers) in ECs who exhibit patterns that differ significantly from those of RCT controls.
In such cases, the adaptive lasso-based approach may achieve  suboptimal performance (as discussed in Section~\ref{motivation}), whereas our proposed approach can effectively accommodate these scenarios.  


We further analyze the properties of the proposed approach, including the convergence rate of the MSE estimator and the validity of the borrowing strategy. 

\begin{theorem}[Convergence rate] \label{thm2} Under Assumption \ref{assump1} and Condition \ref{cond1}, if $|Y|$ is bounded by a finite constant, then for any given borrowed set of ECs $\cS$, 
        \[  \widehat{\textup{MSE}}(\hat{\tau}_{\cS}) - \textup{MSE}(\hat{\tau}_{\cS}) = O_{\P}(1/\sqrt{n}).    \]
\end{theorem}

Theorem \ref{thm2} presents the convergence rate of the proposed MSE estimator. When the sample size of RCT data  is comparable to that of $\cS$ (i.e., $N_{\cE}/N_{\cS}$ converges to a positive constant) or much larger than that of $\cS$ (i.e., $N_{\cE}/N_{\cS}$ diverges to infinity), 
we can write the convergence rate as $O_{\P}(1/\sqrt{N_{\cE}})$. 
If $N_{\cE}$ is much smaller than $N_{\cS}$ (i.e., $N_{\cE}/N_{\cS}$ converges to zero), the convergence rate is $O_{\P}(1/\sqrt{N_{\cS}})$.  
Next, we show the validity of the proposed borrowing strategy. 


\begin{theorem}[Validity] \label{thm3} Under the conditions in Theorem \ref{thm2}, for the given ranking of influence scores, if there exists a positive constant $\eta$,  such that $\textup{MSE}(\hat{\tau}_{\cS^*}) < \textup{MSE}(\hat{\tau}_{\cS_k}) - \eta$ for all $\cS_k \notin \cS^*$, then we have
    \[  \lim_{N_{\mathcal{R}}\to \infty}  \P( \hat \cS \in \cS^* ) = 1.      \]
\end{theorem}

Theorem \ref{thm3} shows that, as the sample size grows, the estimated subset of ECs asymptotically selects an optimal element from 
$\cS^*$ with probability one, provided that the optimal subset is strictly separated from all other candidate subsets of ECs in terms of MSE by a positive margin $\eta$.

\begin{figure*}[t!]
    \centering

    \subfloat[]{
    \begin{minipage}[t]{0.8\linewidth}
    \centering
    \includegraphics[width=1.0\textwidth]{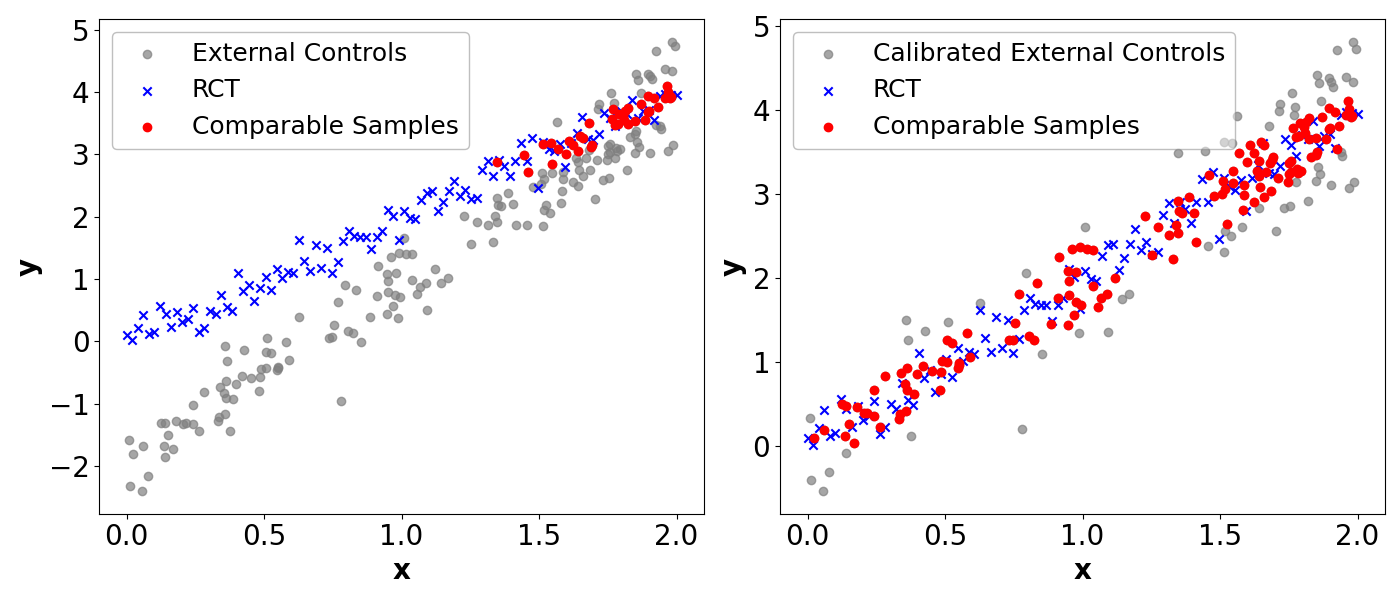}
    \end{minipage}%
    }%
    
    \subfloat[]{
    \begin{minipage}[t]{0.8\linewidth}
    \centering
    \includegraphics[width=1.0\textwidth]{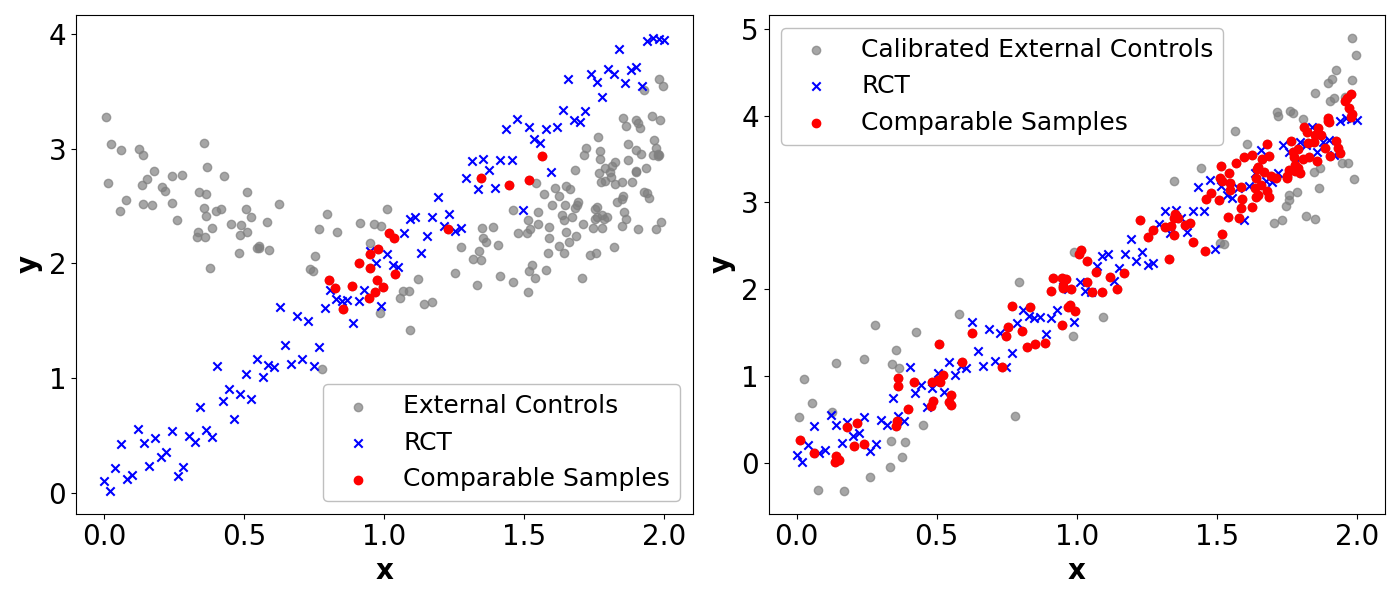}
    \end{minipage}%
    }%
    \centering
    \caption{
   Comparison of the influence-based approach and the calibrated influence-based approach in two simulated cases, see Example 2 in Section~\ref{sec-illustrate} for details. In both cases (a) and (b), the left panel shows the original outcome distribution, featuring a large difference between RCT controls and ECs, while the right panel shows the calibrated distribution, where this difference has been eliminated. Borrowed ECs are marked in red. The left panels in both (a) and (b) demonstrate that the standard influence-based approach borrows only a few ECs. In contrast, the right panels in both (a) and (b) show that, after outcome calibration, the calibrated approach successfully borrows a greater number of comparable ECs.} 
    \label{fig.2}
    \setlength{\abovecaptionskip}{0.cm}
\end{figure*}

\section{Calibrated Sample Borrowing Framework}
\label{sec-calibration}

The influence-based sample borrowing framework developed in Section~\ref{influence-function-framework},  
directly discards less-comparable samples, which may lead to suboptimal data utilization efficiency, especially when the proportion of comparable external samples is low, see Section \ref{sec7-1} for details. 
In this section, we propose a novel calibrated adaptive sample borrowing approach to address this issue and improve data utilization efficiency. 

\subsection{Borrowing Behavior Analysis of the Influence-Based Sample Borrowing Framework} \label{sec7-1}

Before introducing the calibrated approach, we first analyze the borrowing behavior of the influence-based approach. This approach measures comparability by evaluating the influence of each EC sample on the outcome model in RCT controls. It preferentially borrows ECs that exhibit outcome model patterns similar to those of RCT controls.  
When only a few ECs conform to the outcome model patterns of RCT controls (i.e., comparable samples are scarce), the influence-based method can borrow only a limited number of samples from ECs, resulting in an estimator that performs similarly to one based solely on the RCT data. This phenomenon can arise when ECs come from historical data or from different regions, potentially introducing temporal or spatial shifts relative to the patterns observed in RCT controls~\citep{Stuart-Rubin2008}.

We further illustrate this with two simulated cases shown in Fig.~\ref{fig.2}. In case (a) (left panel of Fig.~\ref{fig.2}(a)), only a small fraction of ECs are comparable. In case (b) (left panel of Fig.~\ref{fig.2}(b)), most ECs exhibit patterns that differ substantially from those of RCT controls. In both cases, the influence-based approach cannot effectively leverage external data to improve the estimation of treatment effects in RCTs, resulting in low data-utilization efficiency.




Based on the above analysis, several important questions naturally arise: Can less-comparable samples still be leveraged to improve the estimation of $\tau$? Is it possible to transform some less-comparable samples into comparable ones? Next, we propose a calibration approach to address these questions.
\subsection{Outcome Calibration} 




From the analysis of Fig.~\ref{fig.2}, the low data utilization efficiency arises from systematic differences between the outcome models of RCT controls and ECs. To address this, we propose a novel method for calibrating these differences. 
Specifically, we introduce a bias function $b(x)$ to capture the difference 
and adjust the outcomes of ECs as
\begin{equation*} 
\tilde{Y}_j = Y_j - b(X_j), \quad \text{for all } j \in \mathcal{E},
\end{equation*}
Ideally, to ensure that $b(X)$ captures the systematic difference, we require it to satisfy
$$\mathbb{E}[\tilde{Y} \mid X, R=0, A=0] = \mathbb{E}[Y \mid X, R=1, A=0]$$
which implies that
\[  \mathbb{E}[Y\mid X,R,A=0] = \mu_0(X) + (1- R) b(X).     \]
\begin{algorithm}[t!]
\caption{Adaptive Calibrated Influence-Based Sample Borrowing Approach (ACIB)}
\label{algo:two stage selection}
\begin{algorithmic}[1]
\STATE \textbf{Input:} The RCT data $\{X_i, A_i, Y_i, R_i = 1, i \in \mathcal{R} \}$, and ECs: $(X_j,A_j=0,Y_j,R_j=0,j\in\mathcal{E})$. 
\STATE Stage 1: Bias estimation. 

\begin{itemize}
    \item[$\rhd$] Nuisance parameters estimation. 

\begin{enumerate}
    \item Fit $\hat{m}(X)= \hat{\mathbb{E}}[Y\mid X,A=0]$ using all controls. 
    
    \item Fit $\hat{\pi}_0(X) := \hat{\mathbb{E}}[R\mid X,A=0]$.
    \item Compute residuals: $\hat U = Y - \hat \mu(X),~\hat V=\hat \pi_0(X)-R.$    
\end{enumerate}

    \item[$\rhd$] Bias estimation: estimate bias using Eq. \eqref{eq-bias}. 
\end{itemize}

\STATE Stage 2: Obtain the calibrated ECs: $\{\tilde Z_j = (X_j, \hat{\tilde Y}_j),j\in\mathcal{E}\}$,  
where $\hat{\tilde{Y}}_j = Y_j - \hat{b}(X_j)$.
\STATE Stage 3: Find the optimal subset of calibrated ECs using Algorithm \ref{algo1}.

\STATE \textbf{Return:} 
The comparable calibrated samples set
 $\mathcal{S}_k^*$ and the final estimator of $\tau$.
\end{algorithmic}
\end{algorithm} 
Therefore, to estimate $b(X)$, we could employ all control samples (including both RCT controls and ECs) by modeling $Y$ as a function of $(X, R)$: 
  \begin{equation}\label{eq7-b}
\begin{aligned}
    &Y=\mu_0(X) + (1-R) b(X)+\epsilon,\\
    &\mathbb{E}[\epsilon\mid X,R,A=0]=0.
\end{aligned}
\end{equation}
Motivated by the R-learner framework~\citep{Nie-Wager2021, Wu-Shu-Rlearner}, we first project the model onto the covariate space and obtain that 
\begin{equation}  \label{eq8}
     \mathbb{E}[Y\mid X,A=0]=\mu_0(X)+(1-\pi_0(X))b(X),
\end{equation}
where $\pi_0(X) = \mathbb{E}[R \mid X, A=0]$ denotes the sampling  score among the controls. 
For ease of presentation, let $m(X) = \mathbb{E}[Y \mid X, A=0]$ denote the outcome model for all controls.
Subtracting Eq.~\eqref{eq8} from Eq.~\eqref{eq7-b} yields
\begin{equation} \label{eq9}
Y - m(X) = (\pi_0(X) - R) b(X) + \epsilon. 
\end{equation}
Based on Eq.~\eqref{eq9}, we estimate $b(X)$ as follows: 
\begin{itemize}
    \item Step 1: Estimate $m(X)$ and $\pi_0(X)$ using all controls, denoted by $\hat{m}(X)$ and $\hat{\pi}_0(X)$, respectively. The corresponding residuals are then given by $\hat{U} = Y - \hat{m}(X)$ and $\hat{V} = \hat{\pi}_0(X) - R$.

    \item Step 2: Estimate $b(X)$ by regressing $\hat{U}$ on $\hat{V}$. Without loss of generality, suppose $b(X)$ is parameterized as $b(X; \theta)$ with $\theta$ being its parameter. 
    Then, we estimate $\theta$ by minimizing the penalized least squares loss: 
\begin{equation} \label{eq-bias}
    \sum_{\{i \in \mathcal{R}\cup \mathcal{E}: A_i = 0\}} (\hat{U}_i-b(X_i;\theta)\hat{V}_i)^2+\lambda J(\theta), 
\end{equation}
where $J(\theta)$ is a penalty term used to control the smoothness of $b(X; \theta)$. The estimate of $b(X)$ is denoted by $\hat{b}(X)$.
\end{itemize}

After obtaining the estimate of $b(x)$, we construct the calibrated ECs: 
$\{\tilde Z_j = (X_j, \hat{\tilde Y}_j),j\in\mathcal{E}\}$,  
where $\hat{\tilde{Y}}_j = Y_j - \hat{b}(X_j)$ represents the calibrated outcome. 
Finally, we replace the original ECs with the calibrated ECs and then apply the adaptive influence-based sample borrowing framework to estimate $\tau$, which we refer to as the adaptive calibrated influence-based sample borrowing approach. The associated procedures are summarized in Algorithm~\ref{algo:two stage selection}.

Notably, in Stage 1 of Algorithm~\ref{algo:two stage selection}, we estimate $m(X)$, $\pi_0(X)$, and $b(X)$ using all observed control samples, fully utilizing the available data. 
In Stage 3 of Algorithm~\ref{algo:two stage selection}, we propose to use Algorithm~\ref{algo1} again. The rationale is that, although outcome calibration aligns the conditional means of the potential outcome between RCT controls and ECs, not all  calibrated ECs should be borrowed because some may still introduce substantial noise. By applying the adaptive influence-based sample borrowing framework again, we can selectively retain low-noise samples while filtering out those with high variability.

\section{Borrowing Behavior Analysis of the Proposed Methods} \label{sec-illustrate}
We provide two illustrative examples to intuitively demonstrate the borrowing behavior of the proposed methods, showing how they work and why they are effective.

\begin{example}[Borrowing behavior of the adaptive lasso-based approach and the proposed influence-based method]
We compare the adaptive influence-based borrowing approach with the adaptive lasso-based borrowing approach~\citep{gao2024improving} using a simulated dataset generated as follows: 
\[     Y_{\textup{rt}} = 2 X_{\textup{rt}} + \epsilon_{\textup{rt}},\quad Y_{\textup{ec}} = -0.9 + 2.5 X_{\textup{ec}} + \epsilon_{\textup{ec}},\] 
where $X_{\textup{rt}},X_{\textup{ec}}\sim U(0,2)$,  
  $\epsilon_{\textup{rt}}\sim N(0,0.2^2)$, and $\epsilon_{\textup{ec}}\sim N(0,0.5^2)$.
  The subscript ``\textup{rt}'' denotes the RCT controls sample and ``\textup{ec}'' denotes the EC sample. 
  The sample sizes were $N_c=100$ for RCT controls and $N_{\mathcal{E}}=200$ for the ECs, with the latter including five additional outliers: $(1.6, 0.5)$, $(1.7, 0.5)$, $(1.8, 0.5)$, $(1.9, 0.5)$, and $(2.0, 0.5)$.  

As shown in Fig.~\ref{fig.1}, the adaptive lasso-based approach proposed by~\cite{gao2024improving} achieves suboptimal comparability and is sensitive to outliers. This suggests that relying solely on the conditional mean difference may not be an optimal strategy for identifying comparable samples. In contrast, the proposed influence-based borrowing approach achieves better performance by employing the influence function to quantify the comparability of each EC. 
\end{example}

\begin{example}[Borrowing behavior of the adaptive calibrated influence-based approach] 
We compare the adaptive calibrated influence-based borrowing approach with its non-calibrated counterpart in two simulation cases below. 
\begin{itemize}
    \item Case a: $   Y_{\textup{rt}} = 2 X_{\textup{rt}} + \epsilon_{\textup{rt}},\quad Y_{\textup{ec}} = -2 + 3 X_{\textup{ec}} + \epsilon_{\textup{ec}},$
    \item Case b: $   Y_{\textup{rt}} = 2 X_{\textup{rt}} + \epsilon_{\textup{rt}},\quad Y_{\textup{ec}} = X_{\textup{ec}}^2- 2 X_{\textup{ec}} + 2 + \epsilon_{\textup{ec}},$
\end{itemize}
where $X_{\textup{rt}},X_{\textup{ec}}\sim U(0,2)$, 
  $\epsilon_{\textup{rt}}\sim N(0,0.2^2)$, and $\epsilon_{\textup{ec}}\sim N(0,0.4^2)$. The sample sizes were $N_c=100$ for the RCT controls and $N_{\mathcal{E}}=200$ for the ECs.

As shown in Fig.~\ref{fig.2}, the calibrated approach improves data utilization efficiency relative to the non-calibrated approach by adjusting EC outcomes via subtracting estimated biases from each sample, enabling the borrowing of a substantially larger number of comparable ECs. 
\end{example}


\section{Experiments} \label{sec-experiment}

We conduct experiments on two simulated datasets and one real-world dataset to evaluate the numerical performance of the proposed methods. 


{\bf Baselines.} To show the efficiency of the proposed method for borrowing comparable ECs, we compare our proposed adaptive influence-based borrowing approach (AIB) and adaptive calibrated influence-based borrowing approach (ACIB) with the five baselines:  
\begin{enumerate}
    \item No borrowing approach (NB). Estimating $\tau$ solely based on RCT data using the AIPW estimator~\citep{cao2009improving}.
    \item Full borrowing approach (FB), estimating $\tau$ using the RCT data and all ECs~\citep{li2023improving}. 
    
    \item Full calibrated borrowing approach (FCB). It first calibrates the outcomes in ECs, and then borrows all calibrated ECs using the method of \cite{li2023improving}. This approach is an ablated version of the ACIB method, with the sample borrowing process removed. We include it here for comparison and ablation purposes. 
    
        
    \item Bayesian predictive $p$-value power prior estimator (PPP). It is an extension of the power prior, which discounts each EC sample according to its outcome compatibility using Box's $p$-value~\citep{kwiatkowski2024case}.
    
    \item Adaptive lasso-based borrowing approach (ALB). It evaluates the comparability of ECs using conditional means of outcome differences and borrows only EC samples with estimated zero bias~\citep{gao2024improving}. 
    
\end{enumerate}

The proposed AIB and ACIB approaches assess the comparability of ECs using influence scores, where a smaller score indicates stronger comparability. This contrasts with the ALB approach, which measures comparability using the magnitude of the estimated bias ($\tilde{\bm{b}}$), and the PPP approach, which relies on case weights~\citep{kwiatkowski2024case}. For these approaches, stronger comparability is indicated by a smaller bias estimate and a larger case weight, respectively. By ranking influence scores, bias estimates, or case weights, we can identify the top-$k$ most comparable ECs for the AIB, ACIB, ALB, and PPP approaches, respectively. 
Implementation details for the proposed methods are provided in Appendix C. 

{\bf Evaluation Metrics.} We evaluate the numerical performance using five metrics: point estimate (Est), absolute bias ($|\text{Bias}|$), standard deviation (SD\footnote{The SD is calculated using the variance estimation formulas for $\hat \tau_{\cS}$ given in \ref{section.6.2}.}), mean squared error (MSE). 
In addition, we also present the corresponding number of borrowed ECs ($\#\text{ECs}$).

\subsection{Simulation Study}

For clarity, we list the main research questions (RQs):
\begin{itemize}
    \item \textbf{RQ1:}  What is the overall performance of the proposed approaches versus the competing baselines?

    \item \textbf{RQ2:} How do the MSE curves of different borrowing approaches vary with $K$, when the top-$k$ ECs are borrowed?

    \item \textbf{RQ3:} How effective is the proposed outcome calibration in improving the data utilization efficiency of ECs?

    \item \textbf{RQ4:} How does the degree of outcome difference between RCT controls and ECs influence the performance of the proposed approaches?
   
\end{itemize}



Throughout the simulation, we generate $N = N_{\mathcal{R}} + N_{\mathcal{E}}$ observations, each with $d=8$ dimensional covariates. Each covariate is independently drawn from a truncated normal distribution, $N_{[-2,2]}(0, 1)$. The RCT dataset consists of $N_\mathcal{R}=300$ samples, with $N_t = 200$ in the treated group and $N_c = 100$ in the control group, while the EC dataset includes $N_\mathcal{E} = 1000$ samples.  
 We generate the data source indicator $R$ as $R\mid X\sim \text{Ber}\{\pi(X)\}$ given the sample sizes $(N_{\mathcal{R}},N_{\mathcal{E}})$, where $\pi(X)$ is given as an logistic regression function.    
The treatment $A$ in the RCT is assigned completely at random (i.e., $\mathbb{P}(A_i = 1 \mid R_i = 1) = N_t / N_{\mathcal{R}}$), whereas all ECs consist only of controls (i.e., $\mathbb{P}(A_i = 0 \mid R_i = 0) = 1$).  


\begin{table}[t!]
\caption{Numerical results for various sample borrowing approaches with $\delta = 2.0$. The true ATE is 0.246 for the Linear case and 0.127 for the Nonlinear case.}
\centering
\begin{tabular}{c|cccccc}
\toprule
Linear & Est & $|\text{Bias}|$ & SD  & MSE &$\#\text{ECs}$ \\
\midrule
NB & 0.254  & 0.008  & 0.120  & 0.015   & 0  \\
        FB & 0.147  & 0.099  & 0.083  & 0.017   & 1000  \\
        FCB & 0.201  & 0.045  & 0.081  & 0.009   & 1000  \\ 
        PPP & 0.170  & 0.077  & 0.094  & 0.015   & 157  \\ 
        ALB & 0.219  & 0.027  & 0.095  & 0.010   & 209  \\   \hdashline
        AIB & 0.247  & \textbf{0.001}  & 0.080  & \textbf{0.006}   & 400  \\ 
        ACIB & 0.251  & 0.005  & \textbf{0.078}  & \textbf{0.006}   & 650  \\ 
\midrule

\midrule
NonLinear & Est & $|\text{Bias}|$ & SD  & MSE  &$\#\text{ECs}$ \\
\midrule
NB & 0.106  & 0.021  & 0.120  & 0.015  & 0  \\ 
        FB & -0.095  & 0.223  & 0.096  & 0.059  & 1000  \\ 
        FCB & 0.091  & 0.036  & 0.097  & 0.011  & 1000  \\ 
        PPP & 0.095  & 0.033  & 0.098  & 0.011  & 331  \\ 
        ALB & 0.095  & 0.032  & 0.093  & 0.010  & 201  \\   \hdashline
        AIB & 0.102  & 0.026  & \textbf{0.092}  & 0.009  & 250  \\ 
        ACIB & 0.126  & \textbf{0.001}  & \textbf{0.092}  & \textbf{0.008}  & 450  \\ 
\bottomrule
\end{tabular}
\label{tab2}
\end{table}

Following previous work~\citep{gao2024improving}, we design two data-generating mechanisms for the potential outcomes: a linear outcome model (denoted as ``Linear") and a nonlinear outcome model (denoted as ``Nonlinear"). For the ``Linear" mechanism, the outcomes are generated as follows:  
\begin{align*}
    &Y_{\textup{rt}}=\left\{ \begin{array}{ll} Y_{\textup{rt}}(0)=\beta_1^\intercal X_{\textup{rt}} + \epsilon_{\textup{rt}}, & A = 0 \\ Y_{\textup{rt}}(1)=\beta_1^\intercal X_{\textup{\textup{rt}}} + \alpha_1^\intercal(1, X_{\textup{rt}})  + \epsilon_{\textup{rt}}, & A = 1  \end{array} \right.\\
    &Y_{\textup{ec}}=\beta_1^\intercal X_{\textup{ec}} +\delta \cdot T(X_{\textup{ec}}) + \epsilon_{\textup{ec}}, 
\end{align*}
where the subscript ``$\textup{rt}$'' denotes the RCT sample and ``$\textup{ec}$'' denotes the EC sample, $\beta_1 \sim U([2,3]^d)$ is a $d$-dimensional vector with each component independently and uniformly distributed on $[2,3]$, $\epsilon_{\textup{rt}} \sim N(0,1.0^2)$, $\alpha_1 = 0.1^{d+1}$ (where $0.1^d$ denotes a $(d+1)$-dimensional vector with all elements equal to $0.1$), $\epsilon_{\textup{ec}} \sim N(0,1.2^2)$, and $T(X_{\textup{ec}}) = 0.05^d X_{\textup{ec}}$.

For the ``Linear'' mechanism, the parameter $\delta \geq 0$ controls the magnitude of outcome difference, also known as inconcurrency bias or concept shift, between RCT controls and ECs, with larger $\delta$ indicating more severe inconcurrency bias. 
Similarly, for the ``Nonlinear'' mechanism, the outcomes are generated as follows:  
\begin{align*}
    &Y_{\textup{rt}}=\left\{ \begin{array}{ll}Y_{\textup{rt}}(0)= \textup{exp}\{\beta_2^\intercal X_{\textup{rt}}\} + \epsilon_{\textup{rt}}, &  A = 0 \\ Y_{\textup{rt}}=  \textup{exp}\{\beta_2^\intercal X_{\textup{rt}} + \alpha_2^\intercal(1, X_{\textup{rt}})\}  + \epsilon_{\textup{rt}}, & A = 1  \end{array} \right.\\
    & Y_{\textup{ec}}= \textup{exp}\{\beta_2^\intercal X_{\textup{ec}}\} + \delta\cdot T(X_{\textup{ec}}) + \epsilon_{\textup{ec}},
\end{align*} 
where $\beta_2\sim U([-1,1]^{d})$, $\alpha_1=0.1^{d+1}$, and $T(X_{\textup{ec}})=0.1^{d}X_{\textup{ec}}$. The true ATE in the RCT population ($\tau$) for  ``Linear'' and ``Nonlinear'' cases are 0.246 and 0.127, obtained through a simulation with a sample size of $N_{\mathcal{R}}=100000$.

\begin{figure}[h]
    \centering
    \subfloat[Linear]{
    \begin{minipage}[t]{0.4\linewidth}
    \centering
    \includegraphics[width=1.0\textwidth]{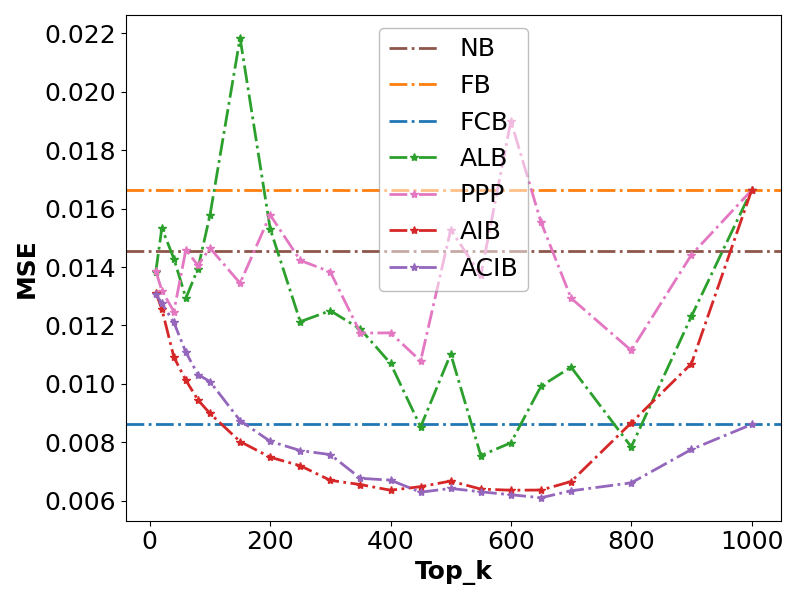}
    \end{minipage}%
    }%
    \subfloat[Nonlinear]{
    \begin{minipage}[t]{0.4\linewidth}
    \centering
    \includegraphics[width=1.0\textwidth]{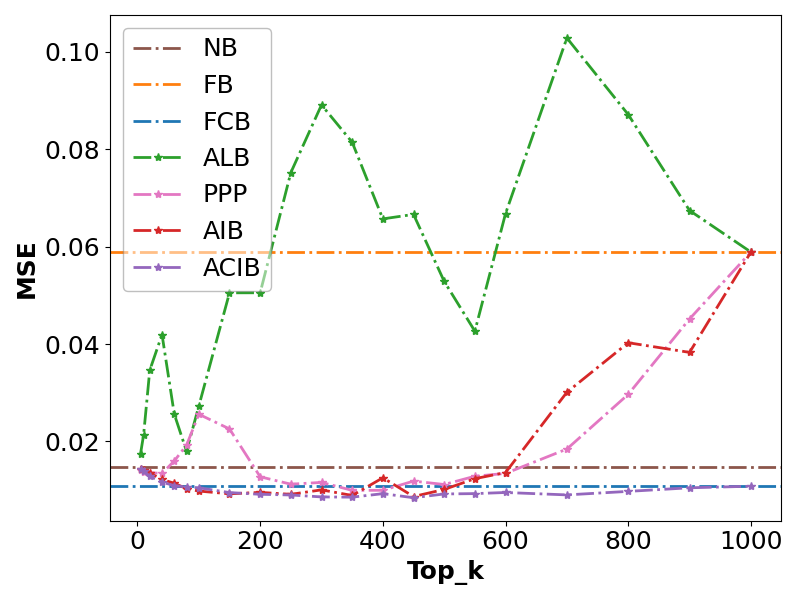}
    \end{minipage}%
    }%
    \caption{Comparison of various approaches
    as top-$k$ ECs are borrowed, where $\delta=2.0$.}
    \label{fig.3}
\end{figure}

\begin{figure}[htbp]
    \centering
    \subfloat[Linear]{
    \begin{minipage}[t]{0.4\linewidth}
    \centering
    \includegraphics[width=1.0\textwidth]{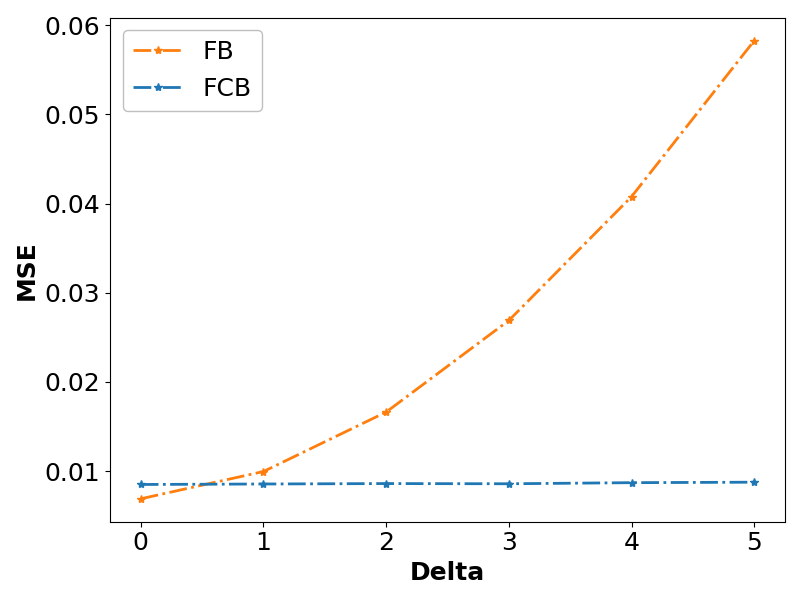}
    \end{minipage}%
    }%
    \subfloat[Nonlinear]{
    \begin{minipage}[t]{0.4\linewidth}
    \centering
    \includegraphics[width=1.0\textwidth]{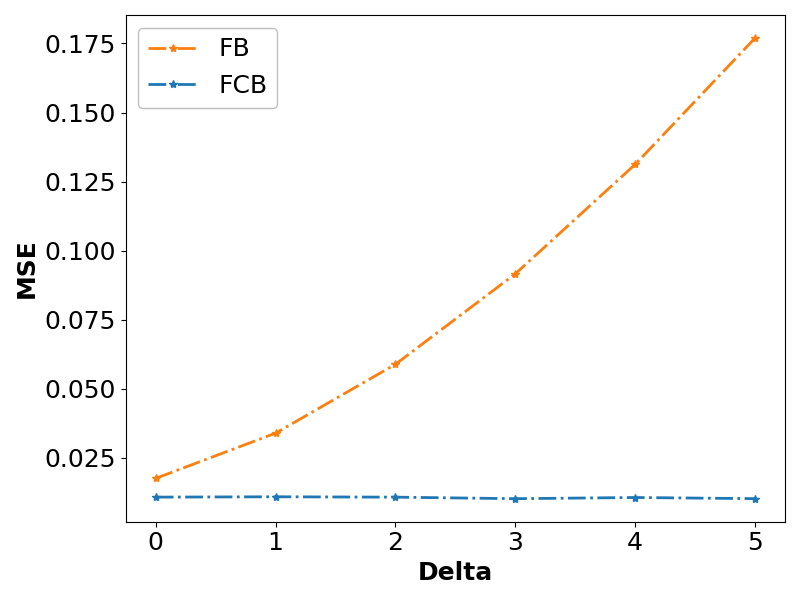}
    \end{minipage}%
    }%
    \caption{Comparison of FB and FCB under different $\delta$.}
    \label{fig.4}
\end{figure}

\textbf{Performance Comparison (RQ1).} 
We first compare the performance of various sample borrowing approaches for estimating $\tau$ under $\delta = 2.0$. The numerical results are summarized in Table~\ref{tab2}, where the borrowed ECs correspond to the optimal MSE. 
From it, we have the following observations. 
(1) The proposed AIB and ACIB outperform baselines, particularly in metrics of $|\text{Bias}|$, SD, and MSE.
This indicates the superiority of the AIB and ACIB in quantifying and borrowing comparable ECs.
 (2) Compared with AIB, ACIB borrows a larger number of ECs in all cases, indicating more comparable ECs are generated by outcome calibration. 

The ACIB contains two key components: outcome calibration and adaptive borrowing. Table~\ref{tab2} also presents an ablation study of the ACIB framework, where AIB is an ablated version obtained by removing the outcome calibration component, FCB is obtained by removing the adaptive borrowing component, and FB is obtained by removing both components. Their relationships are summarized in Table~\ref{tab.abl}.

\begin{table}[H] 
\centering
\caption{Ablation studies}
\resizebox{0.7\linewidth}{!}
{\begin{tabular}{cccc}
\toprule
\text{Method} & Outcome Calibration &  Adaptive Borrowing & Ablative Version   \\
\hline 
ACIB     &   $\checkmark$  &  $\checkmark$  &    $---$  \\  
ACIB    &   $\times$  &  $\checkmark$  & $\Longrightarrow$ AIB{ } \\ 
ACIB     &   $\checkmark$  &  $\times$  & $\Longrightarrow$ FCB \\ 
ACIB     &   $\times$  &  $\times$  & $\Longrightarrow$ FB{\,\, } \\ 
\bottomrule
\end{tabular}}
\label{tab.abl}
\end{table}


 From Table~\ref{tab2}, we observe that FCB consistently outperforms FB, and likewise, ACIB outperforms AIB. These comparisons highlight the critical role of outcome calibration in enhancing data utilization efficiency.
Moreover, AIB outperforms FB, and ACIB further outperforms FCB. These results demonstrate that the influence-based borrowing strategy effectively quantifies the comparability of EC samples and identifies a better subset for borrowing.
\textbf{Sample Borrowing Analysis (RQ2)}. 
We examine how the MSE changes as different top-$k$ ECs are borrowed across various borrowing approaches under $\delta = 2.0$.
As shown in Fig.~\ref{fig.3}, 
the proposed AIB and ACIB initially exhibit a decrease in MSE, followed by an increase as $K$ grows. This pattern indicates that both AIB and ACIB effectively prioritize borrowing highly comparable samples, while the MSE gradually declines once less comparable ECs are included. 
 More results on different $\delta$ are provided in Appendix C.

\textbf{Effectiveness of Outcome Calibration (RQ3).}
We then assess the effectiveness of outcome calibration. 
From Fig.~\ref{fig.3}, the MSE of ACIB reaches its minimum at a larger value of $K$ compared with AIB, indicating that outcome calibration effectively produces more comparable ECs.   
As shown in Fig.~\ref{fig.4} and Fig.~\ref{fig.5}, compared with their non-calibrated counterparts (FB and AIB), the calibrated approaches (FCB and ACIB) remain more stable as the inconcurrency bias ($\delta$) increases. These results demonstrate that outcome calibration mitigates inconcurrency bias and improves data utilization efficiency.

We also assess the effectiveness of the outcome calibration method from the perspectives of distributional distance and the robustness of the estimated influence scores, with the corresponding results provided in Appendix C. 


\textbf{Sensitivity Analysis for Inconcurrency Bias (RQ4).} We further evaluate the performance of the proposed approaches under varying levels of inconcurrency bias ($\delta=\{0,1,2,3,4,5\}$).  
From Figs.~\ref{fig.5}(a) and \ref{fig.5}(c), the number of optimal ECs borrowed by AIB—corresponding to the minimum of the MSE curve—gradually decreases as $\delta$ increases. In particular, when $\delta=0$, AIB borrows the largest set of comparable ECs, which is expected since the true pool of comparable samples shrinks with larger $\delta$. In contrast, Figs.~\ref{fig.5}(b) and \ref{fig.5}(d) show that the optimal number of ECs borrowed by ACIB remains relatively stable across different $\delta$ values, with $\#\text{ECs} \approx 650$ in Linear case and $\#\text{ECs} \approx 450$ in Nonlinear case.


\begin{figure*}[t]
    
    \centering
    \subfloat[Linear, AIB]{
    \begin{minipage}[t]{0.4\linewidth}
    \centering
    \includegraphics[width=1.0\textwidth]{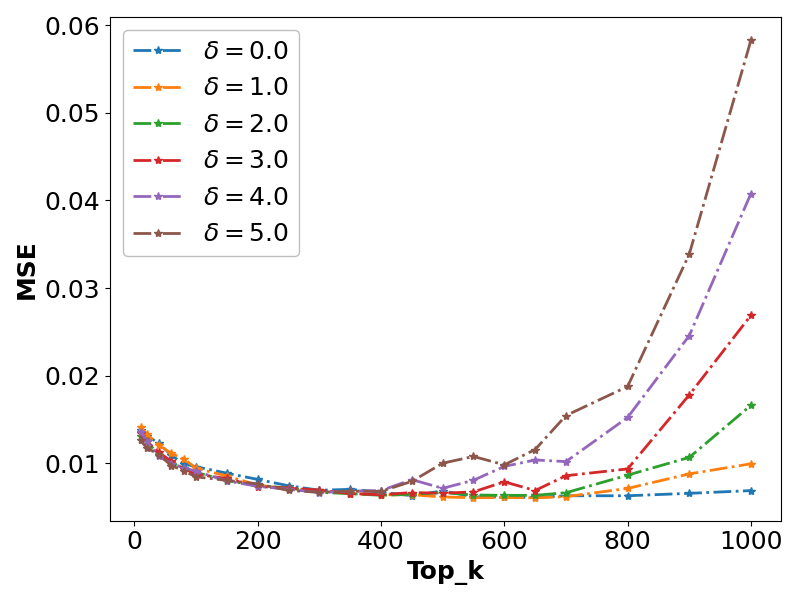}
    \end{minipage}%
    }%
    \subfloat[Linear, ACIB]{
    \begin{minipage}[t]{0.4\linewidth}
    \centering
    \includegraphics[width=1.0\textwidth]{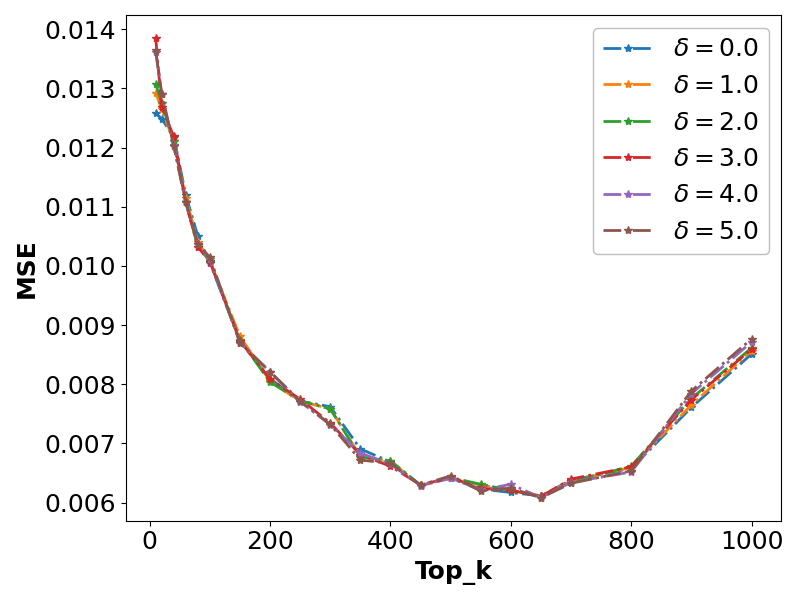}
    \end{minipage}%
    }%
    
    \subfloat[Nonlinear, AIB]{
    \begin{minipage}[t]{0.4\linewidth}
    \centering
    \includegraphics[width=1.0\textwidth]{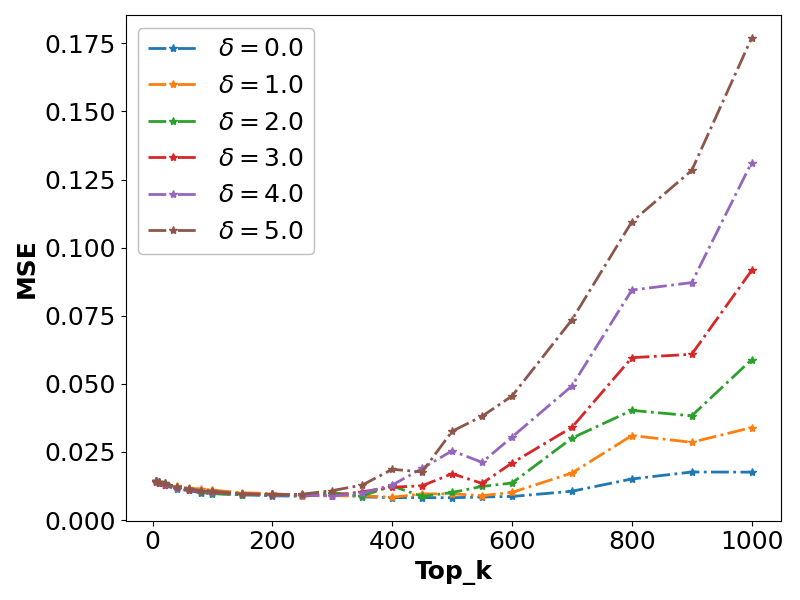}
    \end{minipage}%
    }%
    \subfloat[Nonlinear, ACIB]{
    \begin{minipage}[t]{0.4\linewidth}
    \centering
    \includegraphics[width=1.0\textwidth]{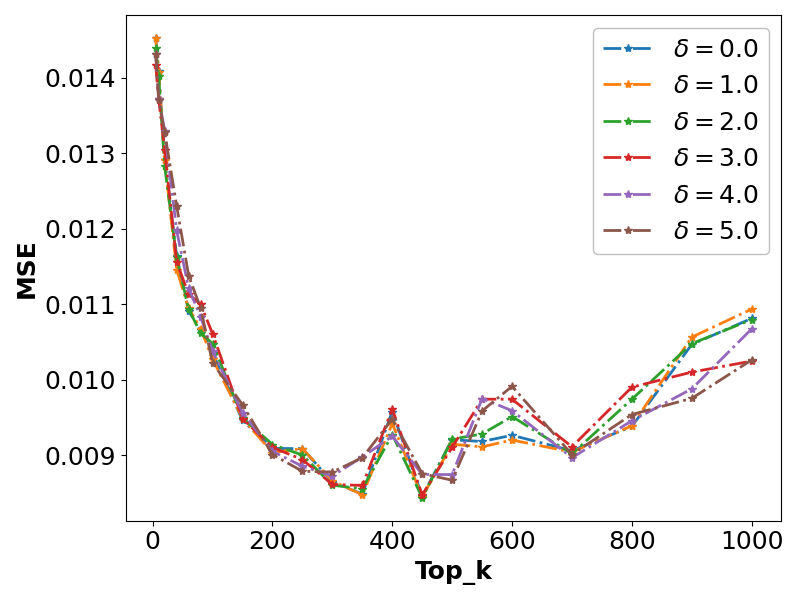}
    \end{minipage}%
    }%
    \caption{Performance of AIB and ACIB approaches at different $\delta$.}
    \label{fig.5}
    
\end{figure*}

\subsection{Real-Data Application}


In addition to simulation studies, we further evaluate the proposed methods using real-world datasets. Specifically, we utilize data from the National Supported Work (NSW) program (representing the RCT dataset)~\citep{lalonde1986evaluating} and the Population Survey of Income Dynamics (PSID) dataset (serving as external controls)~\citep{dehejia2002propensity}\footnote{This data is available at \url{https://users.nber.org/~rdehejia/nswdata2.html}}. The NSW dataset consists of $345$ samples with $185$ in the treated group and $260$ in the control group, while the PSID dataset includes $123$ samples. The NSW program investigates whether providing intensive job training and supported work experience could improve employment outcomes for economically disadvantaged populations. This dataset includes information on training program participation, age, education, race, Hispanic status, marital status, high school degree, and earnings-- RE74 (earnings in 1974), RE75 (earnings in 1975), and RE78 (earnings in 1978)—measured in thousands of dollars. We treat the training program participation as the treatment $A$ and the RE78 as the outcome of interest $Y$ and the remaining eight variables as the baseline covariates $X$. The EC dataset, PSID, is used for external comparison with the NSW dataset, which includes the same covariates and outcome as the NSW dataset, with all observations in the control group ($A = 0$ for all units). Table~\ref{stat info} presents the summary statistics for the baseline covariates and outcome, revealing differences between the two datasets, see Appendix C for details. 

\begin{table*}[htbp]
\caption{Mean and standard deviations (in parentheses) of baseline characteristics.}
\centering
\begin{tabular}{c|cccccc}
\toprule
Variables & Treated Group
in NSW& Control Group
in NSW & Control Group
in PSID\\
\midrule
Age ($X_1$) & 25.82 (7.16) & 25.05 (7.06) & 38.26 (12.89) \\
Education ($X_2$) & 10.35 (2.01) & 10.09 (1.61) & 10.30 (3.18) \\
Race ($X_3$) & 0.84 (0.36) & 0.83 (0.38) & 0.45 (0.50) \\
Hispanic ($X_4$) & 0.06 (0.24) & 0.11 (0.31) & 0.12 (0.32) \\
Marital Status ($X_5$) & 0.19 (0.39) & 0.15 (0.36) & 0.70 (0.46) \\
Nodegree ($X_6$) & 0.71 (0.46) & 0.83 (0.37) & 0.51 (0.50) \\
RE74 ($X_7$) & 2.10 (4.89) & 2.11 (5.69) & 5.57 (7.26) \\
RE75 ($X_8$) & 1.53 (3.22) & 1.27 (3.10) & 2.61 (5.57) \\
RE78 ($Y$) & 6.35 (7.87) & 4.55 (5.48) & 5.28 (7.76) \\
\bottomrule
\end{tabular}
\label{stat info}
\end{table*}

\textbf{Results.}
We investigate whether the job training program increases participants' income in 1978. Fig.~\ref{fig.6} illustrates the sample borrowing behavior of different approaches. Consistent with the findings in the simulation study, both AIB and ACIB show an initial decrease in MSE followed by an increase as $K$ grows (with minor fluctuations), indicating that the proposed methods effectively prioritize borrowing comparable ECs.

\begin{figure}[htbp]
    \centering
    \includegraphics[width=0.5\textwidth]{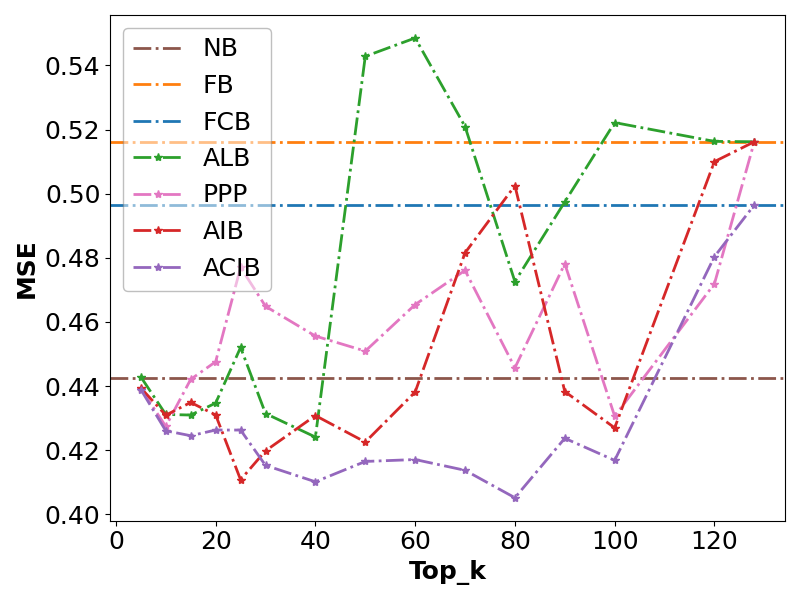}
    \caption{Comparison of various approaches as top-$k$ ECs are borrowed.}
    \label{fig.6}
\end{figure}

Table~\ref{tab5} reports the numerical results, where the borrowed ECs correspond to the optimal MSE. The point estimate obtained from NB is $2.264$, indicating that the job training program yields a significant positive average treatment effect (ATE), increasing income by $2.264$ in 1978. Unlike the simulation study, the true ATE is unknown in this real-world setting. Therefore, we treat $2.264$ as the reference value when computing 
absolute bias ($|\text{Bias}|$) and mean squared error (MSE).

\begin{table}[htbp]
\caption{Analysis results on real-world data (NSW + PSID)}
\centering
\begin{tabular}{c|cccccc}
\toprule
Method & Est & $|\text{Bias}|$ & SD  & MSE &$\#\text{ECs}$ \\
\midrule
NB &2.264 & -- &0.665  & 0.442 &0 \\
FB& 1.942 & 0.322 &0.642 & 0.516 &128\\
FCB&1.971 & 0.293 &0.641 & 0.497 &128 \\
PPP & 2.505  & 0.241  & 0.640  & 0.467  & 28  \\
ALB&2.382 & 0.117 &0.652 & 0.439 &31\\
\hdashline
AIB &2.283  & \textbf{0.018} &0.641 & 0.411 & 25 \\
ACIB &2.347 & 0.082 & \textbf{0.631} & \textbf{0.405} &80 \\
\bottomrule
\end{tabular}
\label{tab5}
\end{table}

From Table~\ref{tab5}, NB exhibits a larger SD than the other approaches, owing to the limited sample size. 
The point estimates of FB, FCB, and PPP differ slightly from NB, likely due to biases introduced by potential outliers in the ECs. In contrast, ALB and the proposed AIB and ACIB yield estimates closer to the benchmark (NB), with AIB and ACIB achieving smaller MSE. 
Furthermore, compared with AIB, ACIB achieves higher estimation efficiency, reflected by smaller SD and MSE. This improvement may stem from the fact that ACIB borrows substantially more comparable ECs ($\#\text{ECs}=80$) than AIB ($\#\text{ECs}=25$). 
In summary, the proposed approaches (AIB and ACIB) achieve effective borrowing from ECs and substantially improve ATE estimation in terms of MSE. 

\section{Conclusion}  \label{sec-conclusion}
In this paper, we first reveal the limitations of existing approaches for borrowing ECs in estimating treatment effects in RCTs. We then propose an adaptive influence-based borrowing framework to overcome these limitations. The framework consists of two key steps: quantifying the comparability of each EC sample and identifying the optimal subset of comparable samples. For the first step, we employ influence functions to compute sample-specific influence scores. For the second step, we develop a data-driven strategy to select the optimal subset of ECs by minimizing the estimator's MSE.Moreover, to address the limitation that influence-based methods may borrow only a small number of ECs when the outcome model in ECs differs substantially from that in RCT control group, we introduce an outcome calibration method to enhance data utilization efficiency and further strengthen the adaptive influence-based borrowing framework.

A limitation of this work is that computing the Hessian matrix for influence scores may have a high computational burden for large-scale models, particularly deep neural networks with millions or even billions of parameters. To address this challenge, the research community has developed several efficient approximation techniques, such as conjugate gradients (CG) for Hessian-vector products~\citep{martens2010deep} and Stochastic estimation methods~\citep{agarwal2017second}. A promising direction for future research would be to integrate these approximations into our framework, thereby enhancing its practicality for complex large-scale models. Another interesting extension is to generalize the proposed method to survival outcomes~\citep{gao2024doubly}.


\bibliographystyle{plainnat}
\bibliography{refs}

@article{rubin1974estimating,
  title={Estimating causal effects of treatments in randomized and nonrandomized studies.},
  author={Rubin, Donald B},
  journal={Journal of educational Psychology},
  volume={66},
  number={5},
  pages={688},
  year={1974},
  publisher={American Psychological Association}
}

@article{Hahn1998,
	author = {Jinyong Hahn},
	journal = {Econometrica},
	number = {2},
	pages = {315-331},
	title = {On the Role of the Propensity Score in Efficient Semiparametric Estimation of Average Treatment Effects},
	volume = {66},
	year = {1998}}

@book{tsiatis2006semiparametric,
  title={Semiparametric theory and missing data},
  author={Tsiatis, Anastasios A},
  volume={4},
  year={2006},
  publisher={Springer}
}

@article{newey1990semiparametric,
  title={Semiparametric efficiency bounds},
  author={Newey, Whitney K},
  journal={Journal of Applied Econometrics},
  volume={5},
  number={2},
  pages={99--135},
  year={1990},
  publisher={Wiley Online Library}
}

@article{splawa1990application,
  title={On the application of probability theory to agricultural experiments. Essay on principles. Section 9},
  author={Splawa-Neyman, Jerzy and Dabrowska, Dorota M and Speed, Terrence P},
  journal={Statistical Science},
  pages={465--472},
  year={1990},
  publisher={JSTOR}
}

@article{rubin1980randomization,
  title={Randomization analysis of experimental data: The Fisher randomization test comment},
  author={Rubin, Donald B},
  journal={Journal of the American Statistical Association},
  volume={75},
  number={371},
  pages={591--593},
  year={1980},
  publisher={JSTOR}
}

@article{rosenbaum1983central,
  title={The central role of the propensity score in observational studies for causal effects},
  author={Rosenbaum, Paul R and Rubin, Donald B},
  journal={Biometrika},
  volume={70},
  number={1},
  pages={41--55},
  year={1983},
  publisher={Oxford University Press}
}

@article{imbens2004nonparametric,
  title={Nonparametric estimation of average treatment effects under exogeneity: A review},
  author={Imbens, Guido W},
  journal={Review of Economics and Statistics},
  volume={86},
  number={1},
  pages={4--29},
  year={2004},
  publisher={MIT Press 238 Main St., Suite 500, Cambridge, MA 02142-1046, USA journals~…}
}

@article{Bang-Robins-2005,
	author = {Heejung Bang and James M. Robins},
	journal = {Biometrics},
	pages = {962-972},
	title = {Doubly robust estimation in missing data and causal inference models},
	volume = {61},
	year = {2005}}

@article{Kang-Schafer-2007,
	author = {Joseph D.Y. Kang and Joseph L. Schafer},
	journal = {Statistical Science},
	pages = {523-539},
	title = {Demystifying double robustness: a comparison of alternative strategies for estimating a population mean from incomplete data},
	volume = {22},
	year = {2007}}

@article{Tan2007,
	author = {Zhiqiang Tan},
	journal = {Statistical Science},
	pages = {560-568},
	title = {Comment: understanding OR, PS and DR},
	volume = {22},
	year = {2007}}

@article{chernozhukov2018double,
	author = {V. Chernozhukov and D. Chetverikov and M. Demirer and E. Duflo and C. Hansen and W. Newey and J. Robins},
	journal = {The Econometrics Journal},
	pages = {1-68},
	title = {Double/debiased machine learning for treatment and structural parameters},
	volume = {21},
	year = {2018}}

@article{wu2024comparative,
  title={On the comparative analysis of average treatment effects estimation via data combination},
  author={Wu, Peng and Luo, Shanshan and Geng, Zhi},
  journal={Journal of the American Statistical Association},
  volume = {0},
 number = {0},
  pages = {1--12}, 
  year={2025}
}

@inproceedings{Kallus-Puli2018,
author = {Nathan Kallus and Aahlad Manas Puli and Uri Shalit},
title = {Removing Hidden Confounding by Experimental Grounding},
year = {2018},
booktitle = {Proceedings of the 32nd International Conference on Neural Information Processing Systems},
numpages = {10911-10920}
}

@article{Colnet-etal2024,
	author = {B{\'e}n{\'e}dicte Colnet and Imke Mayer and Guanhua Chen and Awa Dieng and Ruohong Li and Ga{\"e}l Varoquaux and Jean-Philippe Vert and Julie Josse and Shu Yang},
	journal = {Statistical Science},
	title = {Causal inference methods for combining randomized trials and observational studies: a review},
	volume = {39},
    	pages = {165-191},
	year = {2024}}

@article{Kennedy-2023,
	author = {Edward H. Kennedy},
	journal = {Electronic Journal of Statistics},
	title = {Optimal doubly robust estimation of heterogeneous causal effects},
	year = {2023},
volume = {17},
       pages = {3008--3049}}

@article{Stuart-Rubin2008,
	author = {Elizabeth A. Stuart and Donald B. Rubin},
	journal = {Journal of Educational and Behavioral Statistics},
	title = {Matching With Multiple Control Groups With Adjustment for Group Differences},
	year = {2008},
volume = {33},
       pages = {279--306}}

@InProceedings{Wu-Shu-Rlearner,
  title = 	 {Integrative $R$-learner of heterogeneous treatment effects combining experimental and observational studies},
  author =       {Wu, Lili and Yang, Shu},
  booktitle = 	 {Proceedings of the First Conference on Causal Learning and Reasoning},
  pages = 	 {904--926},
  year = 	 {2022},
  volume = 	 {177},
  publisher =    {PMLR}
}

@article{Nie-Wager2021,
    author = {Nie, X and Wager, S},
    title = {Quasi-oracle estimation of heterogeneous treatment effects},
    journal = {Biometrika},
    volume = {108},
    number = {2},
    pages = {299-319},
    year = {2021}}

@article{kallus2024role,
  title={On the role of surrogates in the efficient estimation of treatment effects with limited outcome data},
  author={Kallus, Nathan and Mao, Xiaojie},
  journal={Journal of the Royal Statistical Society Series B: Statistical Methodology},
  pages={qkae099},
  year={2024},
  publisher={Oxford University Press UK}
}

@article{Degtiar-Rose2023,
	author = {Irina Degtiar and Sherri Rose},
	journal = {Annual Review of Statistics and Its Application},
	pages = {501-524},
	title = {A Review of Generalizability and Transportability},
	volume = {10},
	year = {2023}}

@article{Yang-etal2023,
	author = {Shu Yang and Chenyin Gao and Donglin Zeng and Xiaofei Wang},
	journal = {Journal of the Royal Statistical Society Series B: Statistical Methodology},
	number = {3},
	pages = {575–596},
	title = {Elastic integrative analysis of randomised trial and real-world data for treatment heterogeneity estimatio},
	volume = {85},
	year = {2023}}

@article{Chen-Cai-2021,
	Author = {David Cheng and Tianxi Cai},
	Journal = {arXiv preprint arXiv:2111.15012},
	Title = {Adaptive Combination of Randomized and Observational Data},
	Year = {2021}}

@article{Han-etal-2023,
	Author = {Larry Han and Jue Hou and Kelly Cho and Rui Duan and Tianxi Cai},
	Journal = {Journal of the American Statistical Association},
    number={551},
    pages = {1503–1516},
	Title = {Federated Adaptive Causal Estimation (FACE) of Target Treatment Effects},
    volume={120},
	Year = {2025}}

@article{Dahabreh-etal2019,
	author = {Issa J. Dahabreh and Sarah E. Robertson and Eric J. Tchetgen and Elizabeth A. Stuart and Miguel A. Hernán},
	journal = {Biometrics},
	number = {2},
	pages = {685-694},
	title = {Generalizing causal inferences from individuals in randomized trials to all trial-eligible individuals},
	volume = {75},
	year = {2019}}

@article{Dahabreh-etal2020,
	author = {Issa J. Dahabreh and Sarah E. Robertson and Jon A. Steingrimsson and Elizabeth A. Stuart and Miguel A. Hernán},
	journal = {Statistics in Medicine},
	pages = {1999–2014},
	title = {Extending inferences from a randomized trial to a new
target population},
	volume = {39},
	year = {2020}}

@article{cook1980characterizations,
  title={Characterizations of an empirical influence function for detecting influential cases in regression},
  author={Cook, R Dennis and Weisberg, Sanford},
  journal={Technometrics},
  volume={22},
  number={4},
  pages={495--508},
  year={1980},
  publisher={Taylor \& Francis}
}

@article{shan2022simulation,
  title={A simulation-based evaluation of statistical methods for hybrid real-world control arms in clinical trials},
  author={Shan, Mingyang and Faries, Douglas and Dang, Andy and Zhang, Xiang and Cui, Zhanglin and Sheffield, Kristin M},
  journal={Statistics in Biosciences},
  volume={14},
  number={2},
  pages={259--284},
  year={2022},
  publisher={Springer}
}

@article{agarwal2017second,
  title={Second-order stochastic optimization for machine learning in linear time},
  author={Agarwal, Naman and Bullins, Brian and Hazan, Elad},
  journal={Journal of Machine Learning Research},
  volume={18},
  number={116},
  pages={1--40},
  year={2017}
}

@inproceedings{koh2017understanding,
  title={Understanding black-box predictions via influence functions},
  author={Koh, Pang Wei and Liang, Percy},
  booktitle={Proceedings of the 34th International Conference on Machine Learning},
  pages={1885--1894},
  year={2017},
  organization={PMLR}
}

@inproceedings{martens2010deep,
  title={Deep learning via hessian-free optimization.},
  author={Martens, James and others},
  booktitle={Icml},
  volume={27},
  pages={735--742},
  year={2010}
}

@article{gao2024improving,
  title={Improving randomized controlled trial analysis via data-adaptive borrowing},
  author={Gao, Chenyin and Yang, Shu and Shan, Mingyang and YE, Wenyu and Lipkovich, Ilya and Faries, Douglas},
  journal={Biometrika},
  volume={12},
  number={2},
  pages={asae069},
  year={2025},
  publisher={Oxford University Press}
}

@article{fda2023considerations,
  title={Considerations for the design and conduct of externally controlled trials for drug and biological products guidance for industry},
  author={FDA, US},
  journal={Accessed April},
  volume={1},
  year={2023},
url ={https://www.fda.gov/regulatory-information/search-fda-guidance-documents}
}

@book{imbens2015causal,
  title={Causal inference in statistics, social, and biomedical sciences},
  author={Imbens, Guido W and Rubin, Donald B},
  year={2015},
  publisher={Cambridge University Press}
}

@book{Hampel,
  title={Robust Statistics: The Approach Based on Influence Functions},
  author={Frank R. Hampel and Elvezio M. Ronchetti and Peter J. Rousseeuw and Werner A. Stahel},
  year={2005},
  publisher={ John Wiley \& Sons, Inc.}
}

@article{Wu2023Transfer,
  author = {Lili Wu and Shu Yang},
  title = {Transfer Learning of Individualized Treatment Rules from Experimental to Real-World Data},
  journal = {Journal of Computational and Graphical Statistics},
  volume = {32},
  number = {3},
  pages = {1036--1045},
  year = {2023},
  publisher = {Taylor \& Francis}
}

@article{Qiu-etal2015,
  author = {Sky Qiu and Jens Tarp and Andrew Mertens and Mark van der Laan},
  title = {An Estimator-Robust Design for Augmenting Randomized Controlled Trial with External Real-World Data},
  journal={arXiv preprint arXiv:2501.17835},
  year={2025}
}

@article{athey2017state,
  title={The state of applied econometrics: Causality and policy evaluation},
  author={Athey, Susan and Imbens, Guido W},
  journal={Journal of Economic Perspectives},
  volume={31},
  number={2},
  pages={3--32},
  year={2017},
  publisher={American Economic Association 2014 Broadway, Suite 305, Nashville, TN 37203-2418}
}

@article{lalonde1986evaluating,
  title={Evaluating the econometric evaluations of training programs with experimental data},
  author={LaLonde, Robert J},
  journal={The American Economic Review},
  pages={604--620},
  year={1986},
  publisher={JSTOR}
}

@article{dehejia2002propensity,
  title={Propensity score-matching methods for nonexperimental causal studies},
  author={Dehejia, Rajeev H and Wahba, Sadek},
  journal={Review of Economics and Statistics},
  volume={84},
  number={1},
  pages={151--161},
  year={2002},
  publisher={MIT Press 238 Main St., Suite 500, Cambridge, MA 02142-1046, USA journals~…}
}

@article{cao2009improving,
  title={Improving efficiency and robustness of the doubly robust estimator for a population mean with incomplete data},
  author={Cao, Weihua and Tsiatis, Anastasios A and Davidian, Marie},
  journal={Biometrika},
  volume={96},
  number={3},
  pages={723--734},
  year={2009},
  publisher={Oxford University Press}
}

@article{bojinov2023design,
  title={Design and analysis of switchback experiments},
  author={Bojinov, Iavor and Simchi-Levi, David and Zhao, Jinglong},
  journal={Management Science},
  volume={69},
  number={7},
  pages={3759--3777},
  year={2023},
  publisher={INFORMS}
}

@article{hobbs2011hierarchical,
  title={Hierarchical commensurate and power prior models for adaptive incorporation of historical information in clinical trials},
  author={Hobbs, Brian P and Carlin, Bradley P and Mandrekar, Sumithra J and Sargent, Daniel J},
  journal={Biometrics},
  volume={67},
  number={3},
  pages={1047--1056},
  year={2011},
  publisher={Oxford University Press}
}

@article{stuart2008matching,
  title={Matching with multiple control groups with adjustment for group differences},
  author={Stuart, Elizabeth A and Rubin, Donald B},
  journal={Journal of Educational and Behavioral Statistics},
  volume={33},
  number={3},
  pages={279--306},
  year={2008},
  publisher={Sage Publications Sage CA: Thousand Oaks, CA}
}

@article{Wu-Tan2024,
	author = {Peng Wu and Zhiqiang Tan and Wenjie Hu and Xiao-Hua Zhou},
	journal = {Statistica Sinica},
	title = {Model-Assisted Inference for Covariate-Specific Treatment Effects with High-dimensional Data},
         	Pages = {459--479},
	Volume = {34},
	year = {2024}}

@article{Wu-Han2024,
	Author = {Peng Wu and Shasha Han and Xingwei Tong and Runze Li},
	Journal = {Statistica Sinica},
         	Pages = {747--769},
	Volume = {34},
	Title = {Propensity score regression for causal inference with treatment heterogeneity},
	Year = {2024}}

@article{Wu-Mao2025,
  title={The Promises of Multiple Experiments: Identifying Joint Distribution of Potential Outcomes},
  author={Peng Wu and Xiaojie Mao},
  journal={arXiv preprint arXiv:2504.20470},
  year={2025}
}

@article{hu2023longterm,
  title={Identification and estimation of treatment effects on long-term outcomes in clinical trials with external observational data},
  author={W. Hu and X.H. Zhou and P. Wu},
  journal={Statistica Sinica},
  volume={35},
  pages={1--22},
  year={2025}
}

@article{gao2024doubly,
  title={Doubly protected estimation for survival outcomes utilizing external controls for randomized clinical trials},
  author={Gao, Chenyin and Yang, Shu and Shan, Mingyang and Ye, Wenyu Wendy and Lipkovich, Ilya and Faries, Douglas},
  journal={arXiv preprint arXiv:2410.18409},
  year={2024}
}

@inproceedings{Wu-etal-ShortLong,
author = {Wu, Peng and Shen, Ziyu and Xie, Feng and Wang, Zhongyao and Liu, Chunchen and Zeng, Yan},
title = {Policy learning for balancing short-term and long-term rewards},
year = {2024},
publisher = {PMLR},
booktitle = {Proceedings of the 41st International Conference on Machine Learning
},
articleno = {2206},
pages = {53817-53846},
location = {Vienna, Austria}
}

@inproceedings{yang2024learning,
author = {Yang, Qinwei and Liu, Xueqing and Zeng, Yan and Guo, Ruocheng and Liu, Yang and Wu, Peng},
title = {Learning the optimal policy for balancing short-term and long-term rewards},
year = {2025},
booktitle = {Proceedings of the 38th International Conference on Neural Information Processing Systems},
articleno = {1151},
numpages = {27},
location = {Vancouver, BC, Canada}
}

@article{Wu-etal-2025-Safe,
	Author = {Peng Wu and Qing Jiang and Shanshan Luo},
	Journal = {arXiv preprint arXiv:2505.05308},
	Title = {Safe Individualized Treatment Rules with Controllable Harm Rates},
	Year = {2025}}

@article{neuenschwander2009note,
  title={A note on the power prior},
  author={Neuenschwander, Beat and Branson, Michael and Spiegelhalter, David J},
  journal={Statistics in medicine},
  volume={28},
  number={28},
  pages={3562--3566},
  year={2009},
  publisher={Wiley Online Library}
}

@article{schoenfeld2019design,
  title={Design and analysis of a clinical trial using previous trials as historical control},
  author={Schoenfeld, David Alan and Finkelstein, Dianne M and Macklin, Eric and Zach, Neta and Ennist, David L and Taylor, Albert A and Atassi, Nazem and Pooled Resource Open-Access ALS Clinical Trials Consortium},
  journal={Clinical Trials},
  volume={16},
  number={5},
  pages={531--538},
  year={2019},
  publisher={SAGE Publications Sage UK: London, England}
}

@article{kwiatkowski2024case,
  title={Case weighted power priors for hybrid control analyses with time-to-event data},
  author={Kwiatkowski, Evan and Zhu, Jiawen and Li, Xiao and Pang, Herbert and Lieberman, Grazyna and Psioda, Matthew A},
  journal={Biometrics},
  volume={80},
  number={2},
  pages={ujae019},
  year={2024},
  publisher={Oxford University Press}
}

@article{wang2019using,
  title={Using real-world data to extrapolate evidence from randomized controlled trials},
  author={Wang, Shirley V and Schneeweiss, Sebastian and Gagne, Joshua J and Evers, Thomas and Gerlinger, Christoph and Desai, Rishi and Najafzadeh, Mehdi},
  journal={Clinical Pharmacology \& Therapeutics},
  volume={105},
  number={5},
  pages={1156--1163},
  year={2019},
  publisher={Wiley Online Library}
}

@book{Hernan-Robins2020,
	author = {M.A. Hern{\'a}n and J. M. Robins},
	publisher = {Boca Raton: Chapman and Hall/CRC},
	title = {Causal Inference: What If},
	year = {2020}}

@article{dahabreh2019generalizing,
  title={Generalizing causal inferences from individuals in randomized trials to all trial-eligible individuals},
  author={Dahabreh, Issa J and Robertson, Sarah E and Tchetgen, Eric J and Stuart, Elizabeth A and Hern{\'a}n, Miguel A},
  journal={Biometrics},
  volume={75},
  number={2},
  pages={685--694},
  year={2019},
  publisher={Oxford University Press}
}

@article{li2023improving,
  title={Improving efficiency of inference in clinical trials with external control data},
  author={Li, Xinyu and Miao, Wang and Lu, Fang and Zhou, Xiao-Hua},
  journal={Biometrics},
  volume={79},
  number={1},
  pages={394--403},
  year={2023},
  publisher={Wiley Online Library}
}

@article{pocock1976combination,
  title={The combination of randomized and historical controls in clinical trials},
  author={Pocock, Stuart J},
  journal={Journal of Chronic Diseases},
  volume={29},
  number={3},
  pages={175--188},
  year={1976},
  publisher={Elsevier}
}

@book{van2000asymptotic,
  title={Asymptotic statistics},
  author={Van der Vaart, Aad W},
  volume={3},
  year={2000},
  publisher={Cambridge University Press}
}

@article{valancius2024causal,
  title={A causal inference framework for leveraging external controls in hybrid trials},
  author={Valancius, Michael and Pang, Herbert and Zhu, Jiawen and Cole, Stephen R and Jonsson Funk, Michele and Kosorok, Michael R},
  journal={Biometrics},
  volume={80},
  number={4},
  pages={ujae095},
  year={2024},
  publisher={Oxford University Press}
}

@article{Semenova-Chernozhukov,
	author = {Vira Semenova and Victor Chernozhukov},
	journal = {The Econometrics Journal},
	pages = {264--289},
	title = {Debiased machine learning of conditional average treatment effects and and other causal functions},
	volume = {24},
	year = {2021}}

 \newpage

  \begin{center}
\bf \Large 
Supplementary Material 
\end{center}

\setcounter{equation}{0}
\setcounter{section}{0}
\setcounter{figure}{0}
\setcounter{example}{0}
\setcounter{proposition}{0}
\setcounter{corollary}{0}
\setcounter{theorem}{0}
\setcounter{table}{0}
\setcounter{condition}{0}
\setcounter{lemma}{0}
\setcounter{remark}{0}

\renewcommand {\theproposition} {S\arabic{proposition}}
\renewcommand {\theexample} {S\arabic{example}}
\renewcommand {\thefigure} {S\arabic{figure}}
\renewcommand {\thetable} {S\arabic{table}}
\renewcommand {\theequation} {S\arabic{equation}}
\renewcommand {\thelemma} {S\arabic{lemma}}
\renewcommand {\thesection} {S\arabic{section}}
\renewcommand {\thetheorem} {S\arabic{theorem}}
\renewcommand {\thecorollary} {S\arabic{corollary}}
\renewcommand {\thecondition} {S\arabic{condition}}
\renewcommand {\thepage} {S\arabic{page}}
\renewcommand {\theremark} {S\arabic{remark}}

\setcounter{page}{1}

  \setcounter{equation}{0}
\renewcommand {\theequation} {S\arabic{equation}}
  \setcounter{lemma}{0}
\renewcommand {\thelemma} {S\arabic{lemma}}
   \setcounter{definition}{0}
\renewcommand {\thedefinition} {S\arabic{definition}}
   \setcounter{example}{0}
\renewcommand {\theexample} {S\arabic{example}}
   \setcounter{proposition}{0}
\renewcommand {\theproposition} {S\arabic{proposition}}
   \setcounter{corollary}{0}
\renewcommand {\thecorollary} {S\arabic{corollary}}

 \bigskip 

For clarity, we recall and list the nuisance parameters used in the proposed estimators below.  

 \begin{table}[h!] 
\centering
\caption{Nuisance parameters in $\hat \tau_{\textup{aipw}}$ and $\hat \tau_{\cS}$.} 
\setlength{\tabcolsep}{2pt}
\resizebox{0.6\linewidth}{!}
{\begin{tabular}{ l }
\hline
\text{Nuisance parameters}    \\
\hline 
$e_1(X)=\P(A=1 \mid X,R=1) = \P_{\cS}(A=1 \mid X,R=1)$ \\ 
$\mu_a(X)=\E(Y \mid X, A=a, R=1) = \E_{\cS}(Y \mid X, A=a, R=1)$ \\
   $\pi(X) = \P_{\cS}(R=1\mid X)$ \\ 
$e_{\cS}(X)= \P_{\cS}(A=1 \mid X) = e_1(X) \pi(X)$  \\
$m_a(X) = \E_{\cS} [Y \mid X, A=a]$ \\
\hline 
\end{tabular}} \\
By definition, $\mu_1(X) = m_1(X)$, whereas $\mu_0(X)$ may not equal $m_0(X)$.
\label{tab.2}
\end{table}

In addition, for ease of presentation, we let  $n = N_{\cE} + N_{\cS}$.


\section{Technical Proofs} \label{app:theorem1}


\subsection{Proof of Lemma 1} 

\noindent 
\emph{Proof of Lemma 1}.  Let $f(x)$, $f(x| r=1)$, and $f(x| r = 0)$ represent the density functions of $X$ in the combined data $\P_{\cS}$, RCT data $\P_{\cS}(\cdot | R=1)$ and selected external  data $\P_{\cS}(\cdot | R=0)$, respectively. Denote $f(y_0, y_1 | x)$ as the joint distribution of $(Y(0), Y(1))$ conditional on $X=x$. 


\smallskip 
\emph{First}, we derive the tangent space. The observed data distribution under Assumptions 1 and 3 is given as 
\begin{align*} 
 & p(a, x, y, r) \\  
 ={}&  [ f(a,  y | x)  f(x) \pi(x) ]^r  \cdot [ f(a,  y | x)  f(x) (1- \pi(x)) ]^{1-r}    \\
 ={}& f(x) \times \left[ \left\{ f_1(y|x) e_1(x) \right\}^a \left\{ f_0(y|x) (1 - e_1(x)) \right\}^{1-a} \pi(x) \right]^{r}  \times \left[ f_0(y|x) (1 - \pi(x)) \right]^{1- r} 
\end{align*}
where $f_1(\cdot |x) = \int f(y_0, \cdot | x) dy_0$ and $f_0(\cdot |x) = \int f(\cdot, y_1 | x) dy_1$  are the marginal density of $Y(1)$ and $Y(0)$ conditional on $X=x$, respectively.   
Consider a regular parametric submodel indexed by $\theta$ given as 
\begin{align*}
 p(a, x, y, r; & \theta)  ={} f(x, \theta)  \times \left[ f_1(y|x, \theta)^a f_0(y|x, \theta)^{1-a} \right]^{r} \times  \left[ e_1(x, \theta)^a  (1 - e_1(x, \theta))^{1-a} \pi(x, \theta) \right]^{r} \\
 {}& \times  \left[  f_0(y|x, \theta)  (1 - \pi(x, \theta)) \right]^{1- r}, 
\end{align*}
which equals $p(a, x, y, r)$ when $\theta = \theta_0$. Also, $f_a(y | x, \theta) = f_a(y | x, r=1, \theta) =  f_a(y | x, r=0, \theta)$ by Assumption 3. Then, the score function for this submodel is given by  
 \begin{align*}
	 	s(a, x, y, r; \theta) 
		 ={}& \frac{\partial \log p(a, x, y, r;  \theta)}{\partial \theta}  \\
			={}&  t(x, \theta) + r a \cdot s_1(y|x,\theta) + (1-a) \cdot s_0(y|x,\theta)  + r \frac{a - e_1(x,\theta)}{ e_1(x,\theta)(1 - e_1(x,\theta)) } \dot{e}_1(x, \theta) \\
            {}& +  \frac{r - \pi(x,\theta)}{ \pi(x,\theta)(1 - \pi(x,\theta)) } \dot{\pi}(x, \theta), 
	 \end{align*}	
where  
  \begin{align*} \begin{cases}
    & t(x, \theta) ={} \partial f(x, \theta) / \partial \theta   \\        
    & s_a(y|x,\theta) ={}  \partial \log f_1(y|x,\theta) / \partial \theta  \quad \text{for } a= 0, 1  \\
    & \dot{e}_1(x,\theta) ={} \partial e_1(x,\theta)/\partial \theta \\
    & \dot{\pi}(x,\theta) ={} \partial \pi(x,\theta)/\partial \theta
    \end{cases}
  \end{align*}
Thus, the tangent space is given as    		
  \begin{align*}
         \mathcal{T} ={}&   \{ t(x) + r a s_1(y|x) + (1-a) s_0(y|x) + r (a - e_1(x)) \cdot b_1(x)  + (r - \pi(x)) \cdot b_2(x)     \}, 
  \end{align*}
 where $s_a(y|x)$ satisfies $\E[ s_a(Y|X) \mid X=x ] =  \int s_a(y|x) f_a(y|x) dy = 0$ for $a = 0,1$, $t(x)$ satisfies $\E[ t(X) ] =  \int t(x) f(x)dx = 0$, $b_1(x)$ and $b_2(x)$ are arbitrary square-intergrable measurable function of $x$.
 In addition,  $s_a(y | x) = s_a(y | x, r=1) =  s_a(y | x, r=0)$ according to $f_a(y | x) = f_a(y | x, r=1) = f_a(y | x, r=0) $.

\smallskip 
\emph{Second}, we calculate the pathwise derivative of $\tau$.  
Under the above parametric submodel, the target estimand $\tau = \tau(\theta)$ can be written as  
\begin{align*}
	 \tau(\theta) ={}&  \E[ Y(1) -  Y(0) \mid R = 1 ] \\
        ={}&   \E_{\cS}[ Y(1) -  Y(0) \mid R = 1 ] \\
     ={}& \frac{ \E_{\cS} [ R Y(1) -  R  Y(0)   ] }{ \P_{\cS}(R = 1) }	\\
			  ={}&  \frac{ \int \int  \pi(x,\theta) y f_1(y|x,\theta) f(x,\theta) dy dx    }{  \int \pi(x, \theta) f(x, \theta)  dx } - \frac{\int \int \pi(x,\theta) y f_0(y|x,\theta) f(x,\theta) dy dx    }{  \int \pi(x, \theta) f(x, \theta)  dx }.
	\end{align*}
By calculation, the pathwise derivative of $\tau(\theta)$ at $\theta = \theta_0$ is 
	\begin{align*}
	 \frac{\partial \tau(\theta)}{\partial \theta} \Big|_{\theta = \theta_0} 
={}&  \frac{ \E_{\cS}\Big [ \pi(X) \cdot \E_{\cS}\{ Y(1) \cdot s_1(Y(1) | X) | X \} \Big ] }{q} 
- \frac{ \E_{\cS}\Big [ \pi(X)  \cdot \E_{\cS}\{ Y(0) \cdot s_0(Y(0) | X) | X \} \Big ] }{q} \\
+{}& \frac{ \E_{\cS}\Big [  \Big \{ \pi(X) t(X) + \dot{\pi}(X) \Big \} \cdot \Big \{  m_1(X) -  m_0(X)  - \tau \Big \} \Big ] }{q}.
	\end{align*}

\smallskip 
\emph{Third}, we construct the efficient influence function of $\tau$. Let   
	\begin{align*}
	\phi 
={}& \frac{\pi(X)}{q} \bigg[  \frac{ RA\{Y - m_1(X)\}}{ e_{\mathcal{S}}(X)  } - \frac{(1-A)\{Y - m_0(X)\}   }{1 - e_{\mathcal{S}}(X)} \bigg] + \frac{R}{q}\{ m(X) -  m_0(X) - \tau \}, 
	\end{align*}
 where $e_{\cS}(X) = \P_{\cS}(A=1| X) = \P_{\cS}(A=1|X, R=1) \P_{\cS}(R=1|X) + \P_{\cS}(A=1|X, R=0) \P_{\cS}(R=0|X) = e_1(X)\pi(X)$ and $1 - e_{\cS}(X)= \P_{\cS}(A=0|X) = (1-e_1(X)) \pi(X) + (1-\pi(X))$.

\smallskip 
\emph{Fourth}, we verify that $\tau$ is an influence function of $\tau$. 
This holds if it satisfies the following equation  
	\begin{equation}  \label{A.1}
		 \frac{\partial \tau(\theta) }{\partial \theta}\Big|_{\theta = \theta_0}  =  \E_{\cS}[ \phi  \cdot s(A, X, Y, R; \theta_0) ], 
	\end{equation}
 Next, we give a detailed proof of (\ref{A.1}).  
	\begin{align*}
	   \E_{\cS}[\phi \cdot s(A, X, Y, R; \theta_0) ] 
	={}&  H_1 +  H_2 + H_3,
	\end{align*}
 where 
\begin{align*}
    		H_1 ={}& \E_{\cS} \left [  \left \{  \frac{\pi(X)}{q}  \frac{ R A\{Y - m_1(X)\}}{  e_{\cS}(X)  } \right \}  s(A, X, Y, R; \theta_0) \right ] \\
            H_2 ={}& \E_{\cS} \left [  \left \{ 
	\frac{\pi(X)}{q}  \frac{(1-A)\{Y - m_0(X)\}   }{1 - e_{\cS}(X)} \right \}  s(A, X, Y, R; \theta_0) \right ]   \\
    H_3 ={}& \E_{\cS} \Biggl [  \left \{ \frac{R}{q} (  m_1(X) - m_0(X) - \tau ) \right \}  s(A, X, Y, R; \theta_0) \Biggr ].
\end{align*}
We analyze $H_{1}$, $H_2$, and $H_3$, separately.  
   	   	\begin{align*}
		H_1 
			={}&  \E_{\cS}\left [  \left \{ \frac{\pi(X)}{q}  \frac{ R  A\{Y - m_1(X)\}}{ e_{\cS}(X) }  \right \}   s_1(Y|X)  \right ] \\
			 ={}& \E_{\cS} \left [  \E_{\cS} \left \{ \frac{\pi(X)}{q}  \frac{ R  A\{Y - m_1(X)\}}{ e_{\cS}(X)   }  s_1(Y|X)      \Big | X  \right \}  \right ] \\
			 ={}& \E_{\cS} \Big [  \E_{\cS} \Big \{ \frac{\pi(X)^2}{q}  \frac{  e_1(X) \{Y(1) - m_1(X)\}}{ e_{\cS}(X)   } \times s_1(Y(1) |X)      \Big | X, G=1  \Big \}  \Big ] \\
				 ={}& \E_{\cS} \left [  \frac{\pi(X)}{q}   \E_{\cS} \left \{ \{Y(1) - m_1(X)\}  s_1(Y(1) |X)      \Big | X \right \}  \right ] \\	
				 ={}&  \frac{ \E_{\cS}\Big [ \pi(X) \cdot \E_{\cS}\{ Y(1) \cdot s_1(Y(1) | X) | X \} \Big ] }{q} \\
				 ={}& \text{ the first term of }  \frac{\partial \tau(\theta)}{\partial \theta} \Big|_{\theta = \theta_0}, 
	\end{align*}	
where the fifth equality follows from the fact that $\E[ s_1(Y|X) \mid X=x ] = 0$. 
	\begin{align*}
		H_2 
			={}&  \E_{\cS}\left [  \left \{ 	\frac{\pi(X)}{q}  \frac{(1-A)\{Y - m_0(X)\}   }{ 1 - e_{\cS}(X)}  \right \}  \cdot s_0(Y|X)  \right ] \\
				 ={}& \E_{\cS} \left [ \frac{\pi(X)}{q}  \frac{1- e_{\cS}(X)  }{1 - e_{\cS}(X)}    \cdot  \E_{\cS} \left \{ Y(0) \cdot s_0(Y(0) |X)      \Big | X \right \}  \right ] \\	
                 {}&  - \E_{\cS} \left [ \frac{\pi(X)}{q}  \frac{1- e_{\cS}(X)  }{1 - e_{\cS}(X)}   m_0(X)  \cdot  \E_{\cS} \left \{  s_0(Y(0) |X)      \Big | X \right \}  \right ]  \\ 
				 ={}&  \frac{ \E_{\cS}\Big [ \pi(X)  \cdot \E_{\cS}\{ Y(0) \cdot s_0(Y(0) | X) | X \} \Big ] }{q} \\
				 ={}& \text{ the second term of }  \frac{\partial \tau(\theta)}{\partial \theta} \Big|_{\theta = \theta_0}, 
	\end{align*}	 
and 	
 	\begin{align*}		
		H_3 ={}& \E_{\cS} \Biggl [  \left \{ \frac{R}{q}\{  m_1(X) - m_0(X) - \tau \} \right \}  s(A, X, Y, R; \theta_0) \Biggr ] \\
			={}&   \E_{\cS}\Biggl [  \left \{\frac{R}{q}\{  m_1(X) - m_0(X) - \tau \}  \right \} \Big \{t(X)  \\
            +{}&     R \frac{A - e_1(X)}{ e_1(X)(1 - e_1(X)) } \dot{e}_1(X) +  \frac{R - \pi(X)}{\pi(X)(1 - \pi(X))}  \dot{\pi}(X) \Big \} \Biggr ] \\	 
			  	={}&   \E\Biggl [  \left \{\frac{\pi(X)}{q}\{ m_1(X) - m_0(X) - \tau \}  \right \}   \left \{t(X) +  \frac{\dot{\pi}(X) }{\pi(X)}  \right \} \Biggr ] \\	
			 ={}& \frac{ \E_{\cS}\Big [  \Big \{ \pi(X) t(X) + \dot{\pi}(X) \Big \} \Big \{ m_1(X) - m_0(X)  -  \tau \Big \} \Big ] }{q} \\
			  ={}& \text{ the third term of }  \frac{\partial \tau(\theta)}{\partial \theta} \Big|_{\theta = \theta_0}.
	\end{align*}	
Thus, equation (\ref{A.1}) holds.  

\smallskip 
\emph{Finally}, we show that $\phi$ is efficient influence function by verifying $\phi \in \mathcal{T}$. Let  
  \begin{align*}
      \begin{cases}
          t(X) ={}& \dfrac{\pi(X)}{q}\{ m_1(X) - m_0(X) - \tau \} \\
          s_1(Y|X) ={}& \dfrac{\pi(X)}{q} \frac{(Y - m_1(X))}{ e_{\cS}(X) } \\
          s_0(Y|X) ={}& \dfrac{\pi(X)}{q} \frac{(Y - m_0(X))}{ 1- e_{\cS}(X) } \\
          b_2(X) ={}&  \dfrac{ m_1(X)  - m_0(X)- \tau}{1-q}, 
      \end{cases}
  \end{align*}
  then $\phi$ can be written as 
  	\[ \phi  = t(X) +  R A s_1(Y|X) + (1-A) s_0(Y|X)  + (R- \pi(X)) b_2(X).    \]
Clearly,  $\int s_a(y|x) f_a(y|x) dy = 0$ for $a =0, 1$, and $\int t(x) f(x)dx = 0$, which implies that $\phi \in \mathcal{T}$, and thus $\phi$ is the efficient influence function of $\tau$. 

\hfill $\Box$

\subsection{Proof of Lemma 2} 
 \noindent 
\emph{Proof of Lemma 2.} 
Let $Z = (X, A, R, Y)$, and denote 
\begin{align*} 
	\tilde \phi(Z;  m_0,  m_1, \pi, e_1) ={}&  \frac{\pi(X)}{q} \left[  \frac{ R  A\{Y - m_1(X)\}}{ e_{\cS}(X)  }  - \frac{(1-A)\{Y - m_0(X)\}   }{1 - e_{\cS}(X)  } \right] \\
	        {}&+ \frac{R}{q}\{ m_1(X) - m_0(X) \} 
		\end{align*} 
as the non-centralized efficient influence functions of $\tau$, where $m_0,  m_1, \pi$ and $e_1$ are the nuisance parameters $m_0(x),  m_1(x), \pi(x)$ and $e_1(x)$, respectively. 
Then, the proposed estimator $\hat \tau$ can be written as  
\[ \hat  \tau  =  \frac{1}{n} \sum_{i\in \cE \cup \cS } \tilde  \phi(Z;  \hat m_0, \hat m_1, \hat \pi, \hat e_1), \] 
where $(\hat m_0,  \hat m_1, \hat \pi, \hat e_1)$ are estimates of 
$(m_0,  m_1, \pi, e_1)$. 

We decompose $\hat \tau_{\cS} -  \tau$ as  
$\hat \tau_{\cS} - \tau  =  U_{1n} + U_{2n},$
where 
    \begin{align*}
    	  U_{1n} ={}& \frac{1}{n} \sum_{i\in \cE \cup \cS } [ \tilde \phi(Z_i;  \mu_0,  \mu_1,  \pi,  e_1)  -    \tau ], \\
        U_{2n} ={}& \frac{1}{n} \sum_{i\in \cE \cup \cS } [  \tilde \phi(Z_i;  \hat \mu_0, \hat  \mu_1, \hat \pi, \hat e_1) - \tilde \phi(Z_i;  \mu_0,  \mu_1,  \pi,  e_1)].  
    \end{align*}
Note that $U_{1n}$ is a sum of $n$ independent variables with zero means, and its variance equals $\mathbb{V}^*/n$, where $\mathbb{V}^*$ is the semiparametric efficiency bound of $\tau$.   
By the central limit theorem,  we have 
$\sqrt{n} U_{1n}   \xrightarrow{d} N(0,  \mathbb{V}^* ),$
  where $\xrightarrow{d}$ denotes convergence in distribution.  
 Thus, it suffices to show that $U_{2n} = o_{\P}(n^{-1/2})$.  $U_{2n}$ can be be further decomposed as $ U_{2n} =  U_{2n} - \E_{\cS}[U_{2n}] +  \E_{\cS}[U_{2n}].$ 
   By a Taylor expansion for $\E[U_{2n}]$ yields that  
   \begin{align*}
 \E_{\cS}[U_{2n}] ={}&   \E_{\cS} [  \tilde \phi(Z;  \hat m_0, \hat  m_1, \hat \pi, \hat e_1) - \tilde \phi(Z;  m_0,  m_1,  \pi,  e_1)  ]  \\
     ={}& \partial_{[\hat m_0 - m_0, \hat m_1 - m_1,  \hat \pi - \pi, \hat e_1 - e_1]} \E_{\cS}[  \tilde \phi(Z;  m_0,  m_1,  \pi,  e_1)   ]  \\
     {}& +  \frac 1 2  \partial^2_{[ \hat m_0 - m_0, \hat m_1 - m_1,  \hat \pi - \pi, \hat e_1 - e_1]} \E_{\cS}[  \tilde \phi(Z;  m_0,  m_1,  \pi,  e_1)  ]  \\ 
     {}&+ \cdots
   \end{align*} 
The first-order term  
	\begin{align*}
& \partial_{[\hat m_0 - m_0, \hat m_1 - m_1,  \hat \pi - \pi, \hat e_1 - e_1]} \E_{\cS}[  \tilde \phi(Z;  m_0,  m_1,  \pi,  e_1)   ]  \\
	={}& \E_{\cS} \biggl [ -\frac{1}{q} \left \{ R - \frac{\pi(X)(1-A)}{ 1 - e_{\cS}(X) }    \right \}  (\hat m_0(X) - m_0(X))         \biggr ] \\
     +{}& \E_{\cS} \biggl [ \frac{1}{q} \left \{ R - \frac{\pi(X)RA}{ e_{\cS}(X) }    \right \}  (\hat m_1(X) - m_1(X))       \biggr ]    \\ 
	  +{}&  \E_{\cS} \Big [ \frac{1}{q} \left \{  \frac{ R  A\{Y - m_1(X)\}}{ e_{\cS}(X)  }  + \frac{(1-A)\{Y - m_0(X)\}   }{1- e_{\cS}(X)} \right \}  \times (\hat \pi(X) - \pi(X))         \Big ]    \\ 
	  - {}&   \E_{\cS} \Big [ \frac{\pi(X)}{q} \left \{  \frac{ R A\{Y - m_1(X)\}}{ e_{\cS}(X)^2 }  + \frac{(1-A)\{Y - m_0(X)\}   }{ \{1- e_{\cS}(X)\}^2} \right \}  \times  e_1(X) (\hat \pi(X) - \pi(X))   \Big ]  \\ 
	    -{}&\E_{\cS} \Big [ \frac{\pi(X)}{q} \left \{  \frac{ R A\{Y - m_1(X)\}}{ e_{\cS}(X)^2 }  + \frac{(1-A)\{Y - m_0(X)\}   }{\{1-e_{\cS}(X)\}^2} \right \}  \times \pi(X) (\hat e_1(X) - e_1(X))   \Big ] \\
        ={}& 0. 
	\end{align*}
  where the last equation follows from $\E_{\cS}[ A|X] = e_{\cS}(X)$, $\E_{\cS}[ \pi(X) A \mid X] = \pi(X) e_1(X)$, $\E_{\cS}[ RA(Y - m_1(X)) |X ] = 0$, and $\E_{\cS}[ (1-A)(Y - m_0(X)) |X ] = 0$. 

  For the second-order term, we get 
    \begin{align*}
       & \partial^2_{[ \hat m_0 - m_0, \hat m_1 - m_1,  \hat \pi - \pi, \hat e_1 - e_1]} \E_{\cS}[  \tilde \phi(Z;  m_0,  m_1,  \pi,  e_1)  ]   \\
       ={}&   \E_{\cS} \biggl [ \frac{1}{q} \frac{(1-A)}{1- e_{\cS}(X) }    (\hat m_0(X) - m_0(X)) (\hat \pi(X) - \pi(X))         \biggr ]    \\
		 +{}& \E_{\cS} \biggl [ \frac{1}{q} \frac{ \pi(X) (1-A)}{\{1-  e_{\cS}(X)\}^2}  e_1(X)  (\hat m_0(X) - m_0(X))  \times (\hat \pi(X) - \pi(X))         \biggr ]    \\   
		 +{}& \E_{\cS} \biggl [ \frac{1}{q}   \frac{ \pi(X) (1-A)}{\{1- e_{\cS}(X)\}^2} \pi(X)    (\hat m_0(X) - m_0(X))  \times (\hat e_1(X) - e_1(X))         \biggr ]     \\
			   -{}&  \E_{\cS} \biggl [ \frac{1}{q} \frac{RA}{e_{\cS}(X)}  (\hat m_1(X) - m_1(X))  (\hat \pi(X) - \pi(X))       \biggr ]    \\
       +{}& \E_{\cS} \biggl [ \frac{1}{q}  \frac{\pi(X)RA}{e_{\cS}(X)^2}  e_1(X)  (\hat m_1(X) - m_1(X))  \times (\hat \pi(X) - \pi(X))       \biggr ]    \\ 
      			    +{}& \E_{\cS} \biggl [ \frac{1}{q}  \frac{\pi(X)RA}{e_{\cS}(X)^2}  \pi(X)  (\hat m_1(X) - m_1(X))   \times(\hat e_1(X) - e_1(X))       \biggr ]    \\ 
			     +{}&   \E_{\cS} \biggl [ \frac{1}{q}  \frac{(1-A)  }{e_{\cS}(X)}  (\hat \pi(X) - \pi(X)) (\hat m_0(X) - m_0(X))        \biggr ]    \\ 
			       -{}&  \E_{\cS} \biggl [ \frac{1}{q}   \frac{ R  A}{ e_{\cS}(X)  }   (\hat \pi(X) - \pi(X)) (\hat m_1(X) - m_1(X))         \biggr ]    \\ 
			        	   +{}&   \E_{\cS} \biggl [ \frac{1}{q}   \frac{(1-A)  }{ \{1-e_{\cS}(X)\}^2} e_1(X)  (\hat \pi(X) - \pi(X))  \times (\hat m_0(X) - m_0(X))         \biggr ]  \\    
			   +{}&   \E_{\cS} \biggl [ \frac{1}{q}   \frac{ R  A }{ e_{\cS}(X)^2 }  e_1(X)  (\hat \pi(X) - \pi(X)) (\hat m_1(X) - m_1(X))         \biggr ]  \\     
			     +{}&  \E_{\cS} \biggl [ \frac{1}{q}    \frac{(1-A)   }{\{1-e_{\cS}(X)\}^2}  \pi(X)  (\hat e_1(X) - e_1(X))  \times (\hat m_0(X) - m_0(X))         \biggr ]  \\ 
			        + {}& \E_{\cS} \biggl [ \frac{\pi(X)}{q}   \frac{  A }{ e_{\cS}(X)^2 }  \pi(X)  (\hat e_1(X) - e_1(X))   \times (\hat m_1(X) -m_1(X))          \biggr ]  \\ 
	={}& O_{\P} \biggl (   ||\hat e_1(X) - e_1(X) ||_2  \cdot ( \|  \hat m_1(X) - m_1(X) ||_2    +   ||  \hat m_0(X) - m_0(X) ||_2 )  \\
	{}&  +   ||\hat \pi(X) - \pi(X) ||_2  \cdot ( \|  \hat m_1(X) - m_1(X) ||_2   +   ||  \hat m_0(X) - m_0(X) ||_2 )   \biggr ) \\
	={}& o_{\P}(n^{-1/2}),
     \end{align*}
All higher-order terms can be shown to be dominated by the second-order term. Therefore, 
    $\E_{\cS}[U_{2n}]  = o_{\P}(n^{-1/2}).$  
 In addition,  we get that $U_{2n} - \E_{\cS}[U_{2n}] = o_{\P}(n^{-1/2)}$ by calculating $\text{Var}\{\sqrt{n}(U_{2n} - \E_{\cS}[U_{2n}])\} = o_{\P}(1)$. This proves the conclusion.

\hfill $\Box$

\subsection{Proof of Lemma 3}
\noindent 
\emph{Proof of Lemma 3}. By a standard result of \cite{ Hahn1998, chernozhukov2018double},  
under the conditions in Lemma 2,  
\begin{align*}
 \hat \tau_{\textup{aipw}} ={}& \frac{1}{N_{\mathcal{R}}} \sum_{i\in \mathcal{R}} \left [ \frac{ A_i(Y_i - \mu_1(X_i))}{e_1(X_i)} - \frac{(1-A_i)(Y_i - \mu_0(X_i))}{1 - e_1(X_i)}  \right ] \\
 {}& + \frac{1}{N_{\mathcal{R}}} \sum_{i\in \mathcal{R}} \{ \mu_1(X_i) - \mu_0(X_i)\} + o_{\P}(n^{-1/2}) \\
={}& \frac{1}{n} \sum_{i\in \mathcal{R}\cup\mathcal{S}} \frac{R_i}{q} \left [ \frac{ A_i(Y_i - \mu_1(X_i))}{ e_{1}(X_i) } - \frac{(1-A_i)(Y_i - \mu_0(X_i)) }{1 - e_{1}(X_i)} \right ] \\
{}&+ \frac{1}{n} \sum_{i\in \mathcal{R}\cup\mathcal{S}} \frac{R_i}{q}\{ \mu_1(X_i) - \mu_0(X_i)\} + o_{\P}(n^{-1/2}).
\end{align*}
Thus, the asymptotic variance of $\sqrt{n} \hat \tau_{\textup{aipw}}$, denoted as $\text{asy.var}(\sqrt{n} \hat \tau_{\textup{aipw}})$, is given by
       \begin{align*}
	\text{asy.var}(\sqrt{n} \hat \tau_{\textup{aipw}})  ={}&  \E_\cS \left [  \frac{R}{q^2} \left \{  \frac{ A \text{Var}(Y(1)|X)  }{ e_1(X)^2 }  + \frac{(1-A) \text{Var}(Y(0)|X)  }{ (1 - e_1(X))^2 }     \right \}       \right ] \\
    {}& +  \E_\cS \left [  \frac{R}{q^2}  \{ \mu_1(X) - \mu_0(X) - \tau \}^2  \right ] \\
    ={}&  \E_\cS \left [  \frac{ \pi(X) }{q^2} \left \{  \frac{ \text{Var}(Y(1)|X)  }{ e_1(X) }  + \frac{ \text{Var}(Y(0)|X)  }{ 1 - e_1(X) }     \right \}       \right ] \\
    {}& +  \E_\cS \left [  \frac{R}{q^2}  \{ \mu_1(X) - \mu_0(X) - \tau \}^2  \right ] \\
    ={}&  \E_\cS \left [  \frac{ \pi(X)^2 }{q^2} \left \{  \frac{ \text{Var}(Y(1)|X)  }{ e_{\cS}(X) }  + \frac{  \text{Var}(Y(0)|X)  }{ (1 - e_1(X))\pi(X) }     \right \}       \right ] \\
    {}& +  \E_\cS \left [  \frac{R}{q^2}  \{ \mu_1(X) - \mu_0(X) - \tau \}^2  \right ].
\end{align*}
The asymptotic variance of $\hat \tau_\cS$ is given by  
   \begin{align*}
\text{asy.var}(\sqrt{n} \hat \tau_{\cS}) ={}& \E_\cS[\phi^2]\\ ={}&  \E_\cS \Biggl [  \frac{ \pi(X)^2 }{q^2} \Biggl \{  \frac{ \pi(X) e_1(X) \text{Var}(Y(1)|X)  }{ e_\cS(X)^2 }  + \frac{  (1 - e_{\cS}(X)) \text{Var}(Y(0)|X)  }{ (1 - e_{\cS}(X))^2 }     \Biggr \}       \Biggr ] \\
    {}& +  \E_\cS \left [  \frac{R}{q^2}  \{\mu_1(X) - \mu_0(X) - \tau\}^2  \right ] \\
        ={}&  \E_{\cS}\left [  \frac{ \pi(X)^2 }{q^2} \left \{  \frac{ \text{Var}(Y(1)|X)  }{ e_\cS(X) }  + \frac{  \text{Var}(Y(0)|X)  }{ 1 - e_\cS(X) }     \right \}       \right ]  +   \E_\cS \left [  \frac{R}{q^2}  \{\mu_1(X) - \mu_0(X) - \tau\}^2  \right ]. 
   \end{align*}

 The efficiency gain $\text{asy.var}(\sqrt{n}\hat \tau_{\textup{aipw}}) - \text{asy.var}(\sqrt{n} \hat \tau_{\mathcal{S}})$  is 
       \begin{align*}
          \E_\cS & \Biggl [  \frac{ \pi(X)^2 }{q^2}  \frac{ \text{Var}(Y(0)|X)  }{ (1 - e_1(X))\pi(X) } -     \frac{ \pi(X)^2 }{q^2}  \frac{  \text{Var}(Y(0)|X)  }{ 1- e_\cS(X) }      \Biggr ] \\
           ={}&  \E_\cS \left [  \frac{ \pi(X)^2 }{q^2} \text{Var}(Y(0)|X)  \frac{ 1 - \pi(X)  }{ (1 - e_1(X))\pi(X) \{1- e_\cS(X)  \} }  \right ] \\
           ={}&  \E_\cS \left [  \frac{ \pi(X) }{q^2} \frac{\text{Var}(Y(0)|X)}{ 1 - e_1(X) }  \frac{ 1 - \pi(X)  }{  1- e_\cS(X)  }  \right ]. 
       \end{align*}
    This finishes the proof.

\hfill $\Box$

%
%
%

\subsection{Proof of Theorem 1}
\noindent 
\emph{Proof of Theorem 1}. By the proof of Lemma 2,  if Condition 1 holds,   
we have 
   \begin{align*}
      \hat \tau_{\cS} ={}& \frac{1}{ N_{\mathcal{R}} + N_{\mathcal{S}} } \sum_{i\in \mathcal{R}\cup\mathcal{S}} \frac{\pi(X_i)}{q}   \times \left [  \frac{ R_iA_i(Y_i - m_1(X_i))}{ e_{\mathcal{S}}(X_i)  } - \frac{(1-A_i)(Y_i -  m_0(X_i)) }{1 -  e_{\mathcal{S}}(X_i)} \right ] \\
      +{}&  \frac{1}{ N_{\mathcal{R}} + N_{\mathcal{S}} } \sum_{i\in \mathcal{R}\cup\mathcal{S}} \frac{R_i}{q}\{ m_1(X_i) - m_0(X_i)\}  + o_{\P}(n^{-1/2}).  
   \end{align*}
Note that by definition, $\mu_1(X) = m_1(X)$ but $\mu_0(X) \neq m_0(X)$. Then, under Assumption 1, 
  \begin{align*} 
  \tau ={}& \E_\cS [ \mu_1(X) - \mu_0(X) \mid R=1 ] \\
    ={}& \E_{\cS}\left [    \frac{R}{q}\{ \mu_1(X) - \mu_0(X) \} 
 \right ]  \\
 ={}&   \E_{\cS}\left [    \frac{R}{q}\{ m_1(X) - \mu_0(X) \} 
 \right ].
 \end{align*}
Thus, the bias of $\hat \tau_{\cS}$ is 
   \begin{align*}
   \E_\cS&[\hat \tau_\cS] - \tau \\ 
        ={}& \E_{\cS} \Biggl [ \frac{\pi(X_i)}{q} \left \{  \frac{ R_iA_i(Y_i - m_1(X_i))}{ e_{\mathcal{S}}(X_i)  }  - \frac{(1-A_i)(Y_i -  m_0(X_i)) }{1 -  e_{\mathcal{S}}(X_i)}  \right \}  +  \frac{R_i}{q}\{ m_1(X_i) - m_0(X_i)\} \Biggr ] \\
        {}&- \tau  + o_{\P}(n^{-1/2}) \\
       ={}&  \E_{\cS}\left [ \frac{R}{q}\{ m_1(X) - m_0(X) \} \right ]  -  \E_{\cS} \left [ \frac{R}{q}\{ m_1(X) - \mu_0(X) \} \right ]  + o_{\P}(n^{-1/2})  \\
        ={}& \E_{\cS}\left [ \frac{R}{q} \{ \mu_0(X) - m_0(X)  \} \right ]  + o_{\P}(n^{-1/2}).      
   \end{align*}
In addition, if we denote $\E_{\cS}\left [ \frac{R}{q} \{ \mu_0(X) - m_0(X)  \} \right ]$ by $\textup{bias}(\hat \tau_{\cS})$,  then  
 \begin{align*}
   & \sqrt{n} \{ \hat \tau_{\cS} - \tau - \textup{bias}(\hat \tau_{\cS}) \} \\
    ={}& \frac{1}{\sqrt{n}} \sum_{i\in \cE\cup \cS}  \frac{\pi(X_i)}{q}  \Biggl [  \frac{ R_iA_i(Y_i - m_1(X_i))}{ e_{\mathcal{S}}(X_i)  }  - \frac{(1-A_i)(Y_i -  m_0(X_i)) }{1 -  e_{\mathcal{S}}(X_i)}\Biggr ]  \\
    {}& +  \frac{1}{\sqrt{n }} \sum_{i\in \cE\cup \cS} \left [ \frac{R_i}{q}\{ m_1(X_i) - m_0(X_i)\} - \tau  \right ]   + o_{\P}(n^{-1/2}).
 \end{align*}
This implies the conclusion by central limit theorem.

\hfill $\Box$


\subsection{Proof of Theorem 2}
\noindent 
\emph{Proof of Theorem 2}.  It is sufficient to show that 
      \begin{align}
		 \widehat{\textup{bias}}(\hat \tau_\cS)  -  \textup{bias}(\hat \tau_\cS)  ={}& 	O_{\P}(1/\sqrt{n})  \label{eq-s2} \\
		    \widehat{\textup{var}}(\hat \tau_\cS)  -   \textup{var}(\hat \tau_\cS) ={}& o_{\P}(1/\sqrt{n}).   \label{eq-s3}			
	\end{align}


First, we prove \eqref{eq-s2}. According to the proof of Lemma 2,  
   \begin{align*}
      \hat \tau_{\cS} ={}& \frac{1}{n} \sum_{i\in \mathcal{R}\cup\mathcal{S}} \frac{\pi(X_i)}{q}  \\
     \times {}&  \left [  \frac{ R_iA_i(Y_i - m_1(X_i))}{ e_{\mathcal{S}}(X_i)  } - \frac{(1-A_i)(Y_i -  m_0(X_i)) }{1 -  e_{\mathcal{S}}(X_i)} \right ] \\
      +{}&  \frac{1}{ n} \sum_{i\in \mathcal{R}\cup\mathcal{S}} \frac{R_i}{q}\{ m_1(X_i) - m_0(X_i)\}  + o_{\P}(n^{-1/2}).  
   \end{align*} 
In addition, according to the proof of Lemma 3, 
\begin{align*}
 \hat \tau_{\textup{aipw}} 
={}& \frac{1}{n} \sum_{i\in \mathcal{R}\cup\mathcal{S}} \frac{R_i}{q} \left [ \frac{ A_i(Y_i - \mu_1(X_i))}{ e_{1}(X_i) } - \frac{(1-A_i)(Y_i - \mu_0(X_i)) }{1 - e_{1}(X_i)} \right ] \\
{}&+ \frac{1}{n} \sum_{i\in \mathcal{R}\cup\mathcal{S}} \frac{R_i}{q}\{ \mu_1(X_i) - \mu_0(X_i)\} + o_{\P}(n^{-1/2}),
\end{align*}
where $q = N_{\mathcal{R}}/(N_{\mathcal{R}}+N_{\mathcal{S}})$. 
Since $m_1(X) = \mu_1(X)$,
\begin{align*}
   \hat\tau_{\mathcal{S}} - \hat \tau_{\text{aipw}} 
={}& \frac{1}{n} \sum_{i\in \mathcal{R}\cup\mathcal{S}} \biggl [ \frac{R_i}{q} \frac{(1-A_i)(Y_i - \mu_0(X_i)) }{1 - e_{1}(X_i)}   - \frac{\pi(X_i)}{q} \frac{(1-A_i)(Y_i - m_0(X_i)) }{1 - e_{\mathcal{S}}(X_i)} \biggr ] \\
{}&+ \frac{1}{n } \sum_{i\in \mathcal{R}\cup\mathcal{S}} \frac{R_i}{q}\{ \mu_0(X_i) - m_0(X_i)\} + o_{\P}(n^{-1/2}) \\
:={}& A_{1n} + A_{2n} + o_{\P}(n^{-1/2}),
\end{align*}
We can see that the first term $A_{1n}$ converges to zero at a rate of order $1/\sqrt{n}$, and the second term $A_{2n}$ converges to the bias $\mathbb{E}[ R (\mu_0(X) - m_0(X))/q]$ at a rate of order $1/\sqrt{n}$. Thus, equation \eqref{eq-s2} holds. 

\medskip 
Then, we prove \eqref{eq-s3}. Observe that 
  \[ \widehat{\textup{var}}(\hat \tau_\cS) -   \textup{var}(\hat \tau_\cS)  =  \frac{\hat \sigma^2 - \sigma^2}{n},      \]
it suffices to show that $\hat \sigma^2 - \sigma^2 = o_{\P}(1)$. 
Since  $|Y|$ is bounded by a finite constant, the function  
 \begin{align*}  
 g(X_i, R_i, A_i, Y_i; \hat \pi, \hat m_1, \hat m_0, \hat e_1) 
:=& \frac{\hat \pi(X_i)}{q} \left ( \frac{ R_iA_i(Y_i - \hat m_1(X_i))}{ \hat e_{\mathcal{S}}(X_i)  }  -   \frac{(1-A_i)(Y_i - \hat m_0(X_i)) }{1 - \hat e_{\mathcal{S}}(X_i)}\right )  \\
	     +{}&  \frac{R_i}{q}\{ \hat m_1(X_i) - \hat m_0(X_i)\}
\end{align*} 	     
is Lipschitz continuous with respect to the nuisance parameters $(\pi, m_1, m_0, e_1)$, then the consistency of nuisance parameter estimators implies the consistency of $\hat \sigma^2$.
This finishes the proof of Theorem 2.  

\hfill $\Box$

\subsection{Proof of Theorem 3}

\emph{Proof of Theorem 3.}  
 To show the conclusion 
\[  \lim_{N_{\mathcal{R}}\to \infty}  \P( \hat \cS \in \cS^* ) = 1,      \] 
it suffices to show that 
       \begin{equation}  \label{eq-s4}    \text{MSE}(\hat \tau_{\hat \cS}) - \text{MSE}(\hat \tau_{\cS^*})   \leq  \eta. 
       \end{equation}
This is because, if \eqref{eq-s4} holds and $\hat \cS \notin \cS^*$, it contradicts the condition  ``$\textup{MSE}(\hat{\tau}_{\cS^*}) < \textup{MSE}(\hat{\tau}_{\cS_k}) - \eta$ for all $\cS_k \notin \cS^*$", and thus we must have $\hat \cS \in \cS^*$. 


 Next, we prove \eqref{eq-s4}. We consider a decomposition of $ \text{MSE}(\hat{\tau}_{\hat\cS}) - \text{MSE}(\hat{\tau}_{\cS^*})$, 
    \begin{align*}
        \text{MSE}(\hat{\tau}_{\hat \cS}) - \text{MSE}(\hat{\tau}_{\cS^*})  
    ={}&  \underbrace{\text{MSE}(\hat{\tau}_{\hat \cS}) -  \widehat{\text{MSE}}(\hat{\tau}_{\hat \cS})}_{B_{1n}} +   \underbrace{\widehat{\text{MSE}}(\hat{\tau}_{\hat \cS})  - \widehat{\text{MSE}}(\hat{\tau}_{ \cS^*})}_{B_{2n}}  \\
    {}& + \underbrace{\widehat{\text{MSE}}(\hat{\tau}_{\cS^*})  -     \text{MSE}(\hat{\tau}_{\cS^*})}_{B_{3n}}.  
    \end{align*}
Since $\hat \cS  =  \arg \min_{\mathcal{S}_k  \in {\bf S} } \widehat{\text{MSE}}(\hat{\tau}_{\cS_k})$, $B_{2n}\leq 0$, which implies that 
   \[    \text{MSE}(\hat{\tau}_{\hat \cS}) - \text{MSE}(\hat{\tau}_{\cS^*})    \leq B_{1n} + B_{3n}.              \]

By Theorem 2,  $\widehat{\text{MSE}}(\hat{\tau}_{\cS}) - \text{MSE}(\hat{\tau}_{\cS}) = o_{\P}(1)$ for $\cS = \hat \cS, \cS^*$. Thus, for the above given $\eta$,  there exists a positive interger $M$, such that when the sample size $N_{\cE} + N_{\cS}$ is large than $M$,  we have 
      \begin{align*}
   |B_{1n}| ={}&  |\widehat{\text{MSE}}(\hat{\tau}_{\hat \cS}) - \text{MSE}(\hat{\tau}_{\hat \cS})| \leq \eta/2,  \\    
   |B_{3n}| ={}& |\widehat{\text{MSE}}(\hat{\tau}_{\cS^*}) - \text{MSE}(\hat{\tau}_{\cS^*})| \leq \eta/2.
      \end{align*}
Therefore, the inequality \eqref{eq-s4} holds. This completes the proof. 

\hfill $\Box$



\section{Regularity Conditions} 

\begin{assumption} \label{assump-A1}  
Let $Z = (X, A, Y, R)$, and write $\phi$ in Lemma 1 as $\phi(Z)$, 
suppose that $\phi(Z)$ belongs to a Donsker class~\citep{van2000asymptotic}.
\end{assumption}

Assumption \ref{assump-A1} regularizes the complexity of the functional space of the nuisance parameters~\citep{gao2024improving}, and it may not be required if a cross-fitting procedure is used for estimating the nuisance parameters~\citep{chernozhukov2018double}.

\section{Additional Experiments} \label{app:more_exps}

\subsection{Implementation Details}

We use logistic regression to estimate the propensity scores $e_1(X)$ and sampling score $\pi(X)$. For the outcome model, we use linear regression in the "Linear" simulation and exponential regression in the "Nonlinear" simulation to model the expected outcome given the treatment and covariates. 
For real-world data, we train the outcome model of the RCT data or external control through a Multi-Layer Perceptron (MLP). All experiments were run on an NVIDIA GeForce RTX 3090 to ensure efficient computation.

\subsection{Performance Comparison at $\delta=\{1.0,3.0\}$}
\label{addc_a}
This section provides further experimental results, including a detailed comparison of performance and sample borrowing behavior across methods for $\delta=\{1,3\}$, as shown in Tables~\ref{tab2_add},~\ref{tab3_add} and Figures~\ref{fig.1_add},~\ref{fig.2_add}. As the results indicate, the proposed AIB and ACIB approaches achieve consistent and superior performance over the existing baselines.

\begin{table}[ht!]
\caption{Numerical results for various sample borrowing approaches with $\delta = 1.0$. The true ATE is 0.246 for the Linear case and 0.127 for the Nonlinear case.}
\centering
\begin{tabular}{c|cccccc}
\toprule
Linear & Est & $|\text{Bias}|$ & SD  & MSE  &$\#\text{ECs}$ \\
\midrule
NB & 0.254  & 0.008  & 0.120  & 0.015   & 0  \\
        FB & 0.189  & 0.057  & 0.082  & 0.010   & 1000  \\
        FCB & 0.202  & 0.044  & 0.081  & 0.009  & 1000  \\ 
        PPP & 0.190  & 0.057  & 0.093  & 0.012  & 155  \\
        ALB & 0.240  & 0.006  & 0.093  & 0.009  & 259  \\ 
        \hdashline
        AIB & 0.239  & 0.007  & \textbf{0.078}  & \textbf{0.006}  & 600  \\
        ACIB & 0.252  & \textbf{0.006}  & \textbf{0.078}  & \textbf{0.006}  & 650  \\
\midrule

\midrule
NonLinear & Est & $|\text{Bias}|$ & SD  & MSE &$\#\text{ECs}$ \\
\midrule
        NB & 0.106  & 0.021  & 0.120  & 0.015  & 0  \\ 
        FB & -0.030  & 0.157  & 0.096  & 0.034  & 1000  \\
        FCB & 0.090  & 0.038  & 0.097  & 0.011  & 1000  \\
        PPP & 0.136  & 0.008  & 0.096  & 0.009  & 405  \\
        ALB & 0.004  & 0.124  & 0.119  & 0.029  & 209  \\
        \hdashline
        AIB & 0.116  & 0.011  & \textbf{0.091}  & \textbf{0.008}  & 400  \\
        ACIB & 0.125  & \textbf{0.003}  & 0.092  & \textbf{0.008}  & 450  \\
\bottomrule

\end{tabular}
\label{tab2_add}
\end{table}




\begin{table}[ht!]
\caption{Numerical results for various sample borrowing approaches with $\delta = 3.0$. The true ATE is 0.246 for the Linear case and 0.127 for the Nonlinear case.}
\centering
\begin{tabular}{c|cccccc}
\toprule
Linear & Est & $|\text{Bias}|$ & SD  & MSE &$\#\text{ECs}$ \\
\midrule
NB & 0.254  & 0.008  & 0.120  & 0.015  & 0  \\ 
        FB & 0.105  & 0.141  & 0.084  & 0.027  & 1000  \\ 
        FCB & 0.202  & 0.045  & 0.081  & 0.009  & 1000  \\
        PPP & 0.166  & 0.080  & 0.094  & 0.015  & 162  \\ 
        ALB & 0.199  & 0.047  & 0.097  & 0.012  & 205  \\ 
        \hdashline
        AIB & 0.236  & 0.010  & 0.080  & \textbf{0.006}  & 400  \\ 
        ACIB & 0.251  & \textbf{0.005}  & \textbf{0.078}  & \textbf{0.006}  & 650  \\
\midrule

\midrule
NonLinear & Est & $|\text{Bias}|$ & SD  & MSE &$\#\text{ECs}$ \\
\midrule
NB & 0.106  & 0.021  & 0.120  & 0.015  & 0  \\ 
        FB & -0.159  & 0.286  & 0.098  & 0.092  & 1000  \\
        FCB & 0.100  & 0.027  & 0.097  & 0.010  & 1000  \\ 
        PPP & 0.014  & 0.113  & 0.101  & 0.023  & 218  \\ 
        ALB & 0.016  & 0.111  & 0.130  & 0.029  & 175  \\ 
        \hdashline
        AIB & 0.105  & 0.023  & \textbf{0.092}  & 0.009  & 300  \\ 
        ACIB & 0.134  & \textbf{0.006}  & \textbf{0.092}  & \textbf{0.008} & 450  \\ 
\bottomrule

\end{tabular}
\label{tab3_add}
\end{table}

\begin{figure}[h]
    \centering
    \subfloat[Linear]{
    \begin{minipage}[t]{0.4\linewidth}
    \centering
    \includegraphics[width=1.0\textwidth]{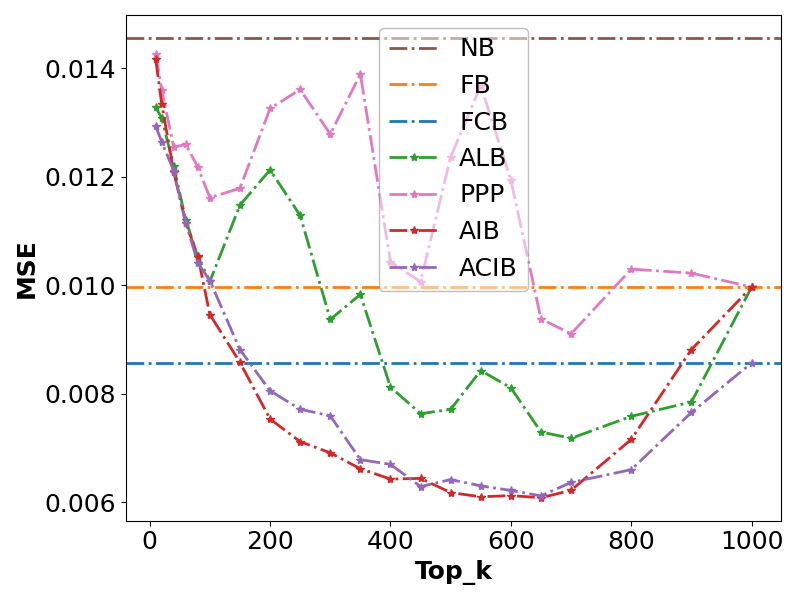}
    \end{minipage}%
    }%
    \subfloat[Nonlinear]{
    \begin{minipage}[t]{0.4\linewidth}
    \centering
    \includegraphics[width=1.0\textwidth]{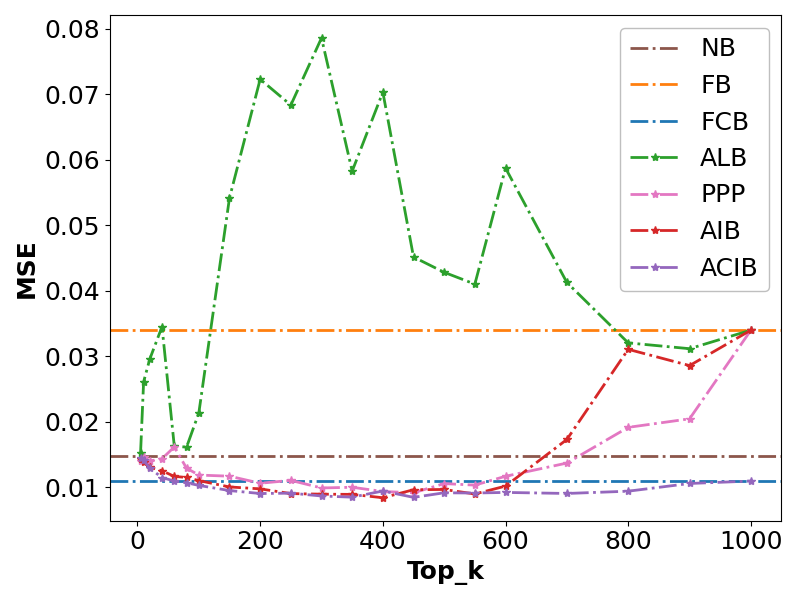}
    \end{minipage}%
    }%
    \caption{Comparison of various approaches
    as top-$k$ ECs are borrowed, where $\delta=1.0$.}
    \label{fig.1_add}
\end{figure}

\begin{figure}[h]
    \centering
    \subfloat[Linear]{
    \begin{minipage}[t]{0.4\linewidth}
    \centering
    \includegraphics[width=1.0\textwidth]{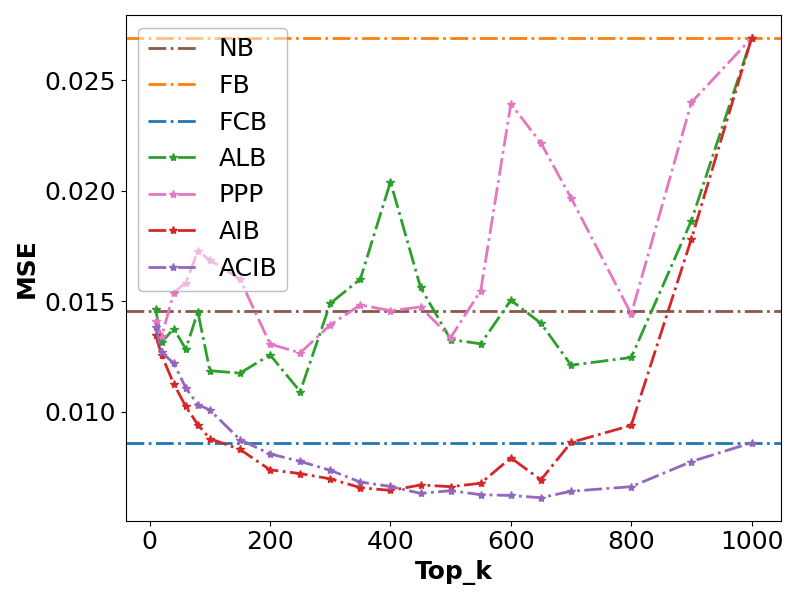}
    \end{minipage}%
    }%
    \subfloat[Nonlinear]{
    \begin{minipage}[t]{0.4\linewidth}
    \centering
    \includegraphics[width=1.0\textwidth]{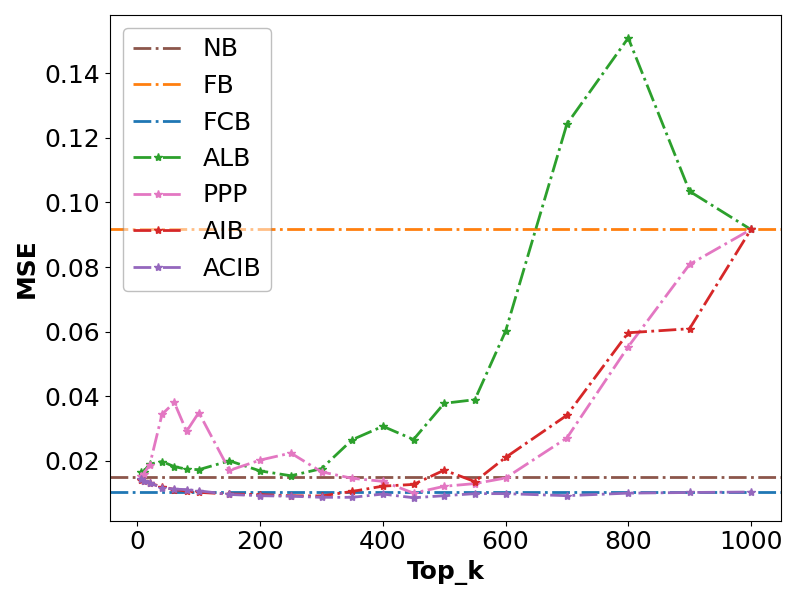}
    \end{minipage}%
    }%
    \caption{Comparison of various approaches
    as top-$k$ ECs are borrowed, where $\delta=3.0$.}
    \label{fig.2_add}
\end{figure}

\subsection{Some Results about outcome calibration}

In this section, we demonstrate the effectiveness of the proposed outcome calibration method. Our goal is to calibrate the outcomes in ECs so that the conditional means of the potential outcomes under control are aligned between the RCT controls and ECs. We define the ideal calibrated outcome in the ECs as $$Y_{\text{ideal}}=\hat \mu_0(X_j), j\in\mathcal{E}.$$ To quantify similarity between outcome distributions, we employ the Euclidean distance metric.

We define the diatance between $Y_{\text{ideal }}$ and the calibrated outcome $\tilde Y$ as
$$
d(Y_{\text{ideal}}, \tilde Y)=\sqrt{\sum_{j=1}^{N_{\mathcal{E}}}(Y_{\text{ideal},j}-\tilde Y_j)^2}.
$$
Similarly, the distance between $Y_{\text{ideal }}$ and the original external-control outcomes $Y$ is
$$
d(Y_{\text{ideal}},  Y)=\sqrt{\sum_{j=1}^{N_{\mathcal{E}}}(Y_{\text{ideal},j}- Y_j)^2}.
$$
Larger values of $d(Y_{\text{ideal}}, \tilde Y)$ or $d(Y_{\text{ideal}}, Y)$ indicate greater discrepancy (and thus lower similarity) between the corresponding outcome vectors.

Figure~\ref{fig.3_add} reports the results for all three datasets, including two simulated datasets and one real dataset. As shown in Figure \ref{fig.3_add}(a) and (b), as $\delta$ increases, $d(Y_{\text{ideal}}, Y)$ increases steadily, whereas $d(Y_{\text{ideal}}, \tilde Y)$ remains relatively stable and consistently smaller than $d(Y_{\text{ideal}}, Y)$. This pattern suggests that the calibration method effectively eliminates the outcome inconcurrency bias. As shown in Figure \ref{fig.3_add}(c), the outcome calibration likewise eliminates the outcome inconcurrency bias in the real dataset.

\begin{figure*}[htbp]
    \centering
    \subfloat[Linear]{
    \begin{minipage}[t]{0.33\linewidth}
    \centering
    \includegraphics[width=1.0\textwidth]{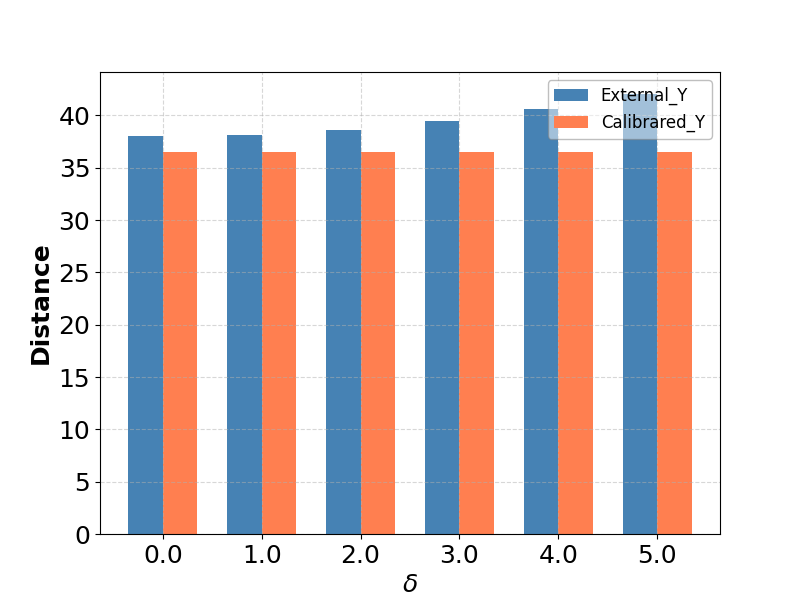}
    \end{minipage}%
    }%
    \subfloat[Nonlinear]{
    \begin{minipage}[t]{0.33\linewidth}
    \centering
    \includegraphics[width=1.0\textwidth]{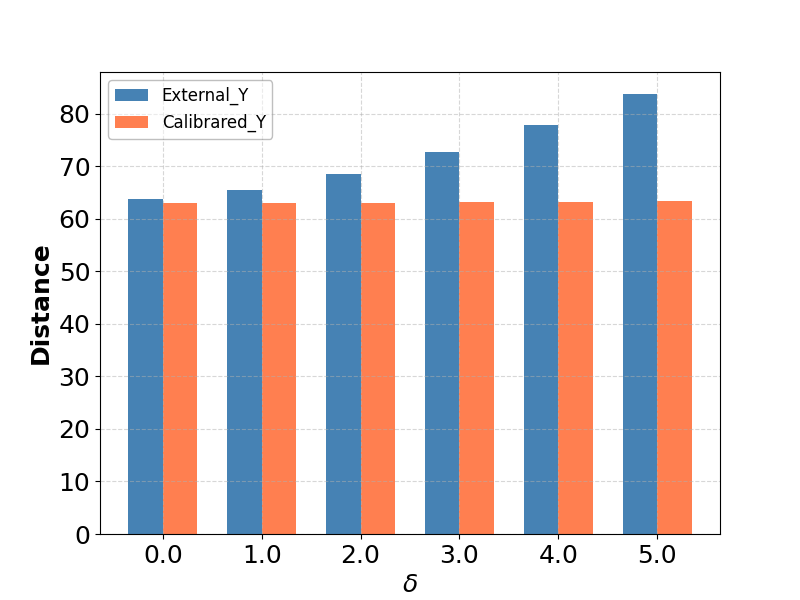}
    \end{minipage}%
    }%
    \subfloat[NSW $\&$ PSID]{
    \begin{minipage}[t]{0.33\linewidth}
    \centering
    \includegraphics[width=1.0\textwidth]{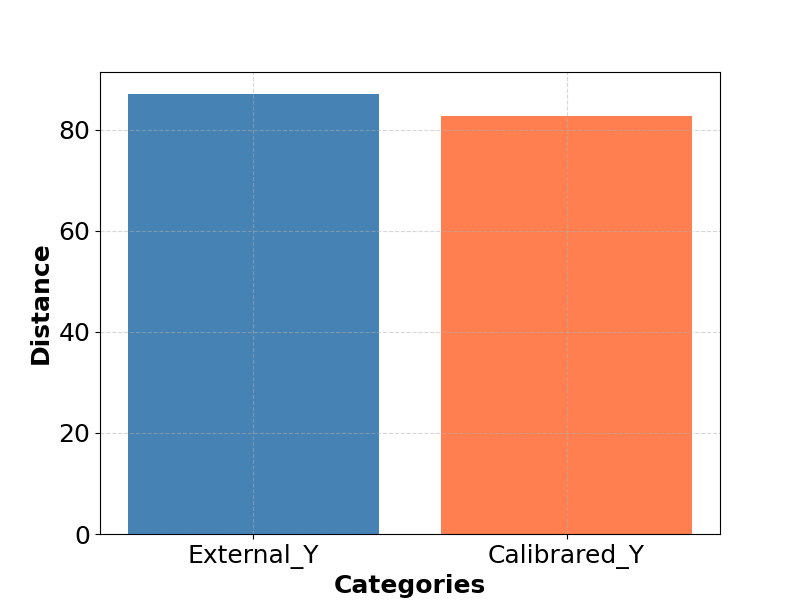}
    \end{minipage}%
    }%
    \caption{Distance between the calibrated outcome and the ideal outcome.}
    \label{fig.3_add}
\end{figure*}

\begin{figure*}[htbp]
    \centering
    \subfloat[Linear]{
    \begin{minipage}[t]{0.33\linewidth}
    \centering
    \includegraphics[width=1.0\textwidth]{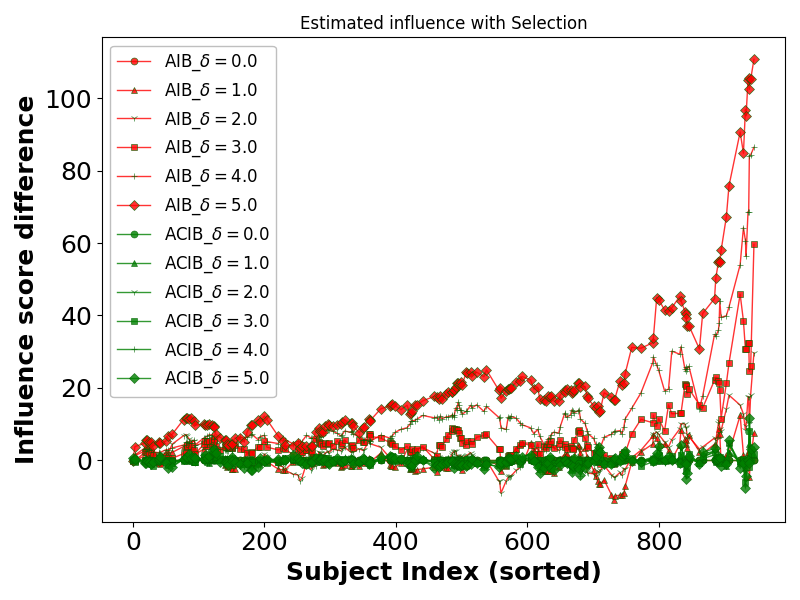}
    \end{minipage}%
    }%
    \subfloat[Nonlinear]{
    \begin{minipage}[t]{0.33\linewidth}
    \centering
    \includegraphics[width=1.0\textwidth]{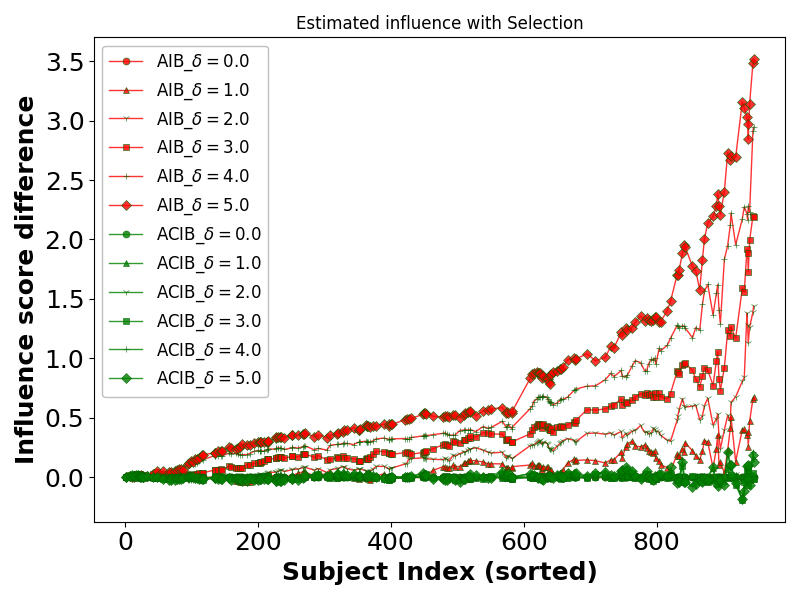}
    \end{minipage}%
    }%
    \subfloat[NSW $\&$ PSID]{
    \begin{minipage}[t]{0.33\linewidth}
    \centering
    \includegraphics[width=1.0\textwidth]{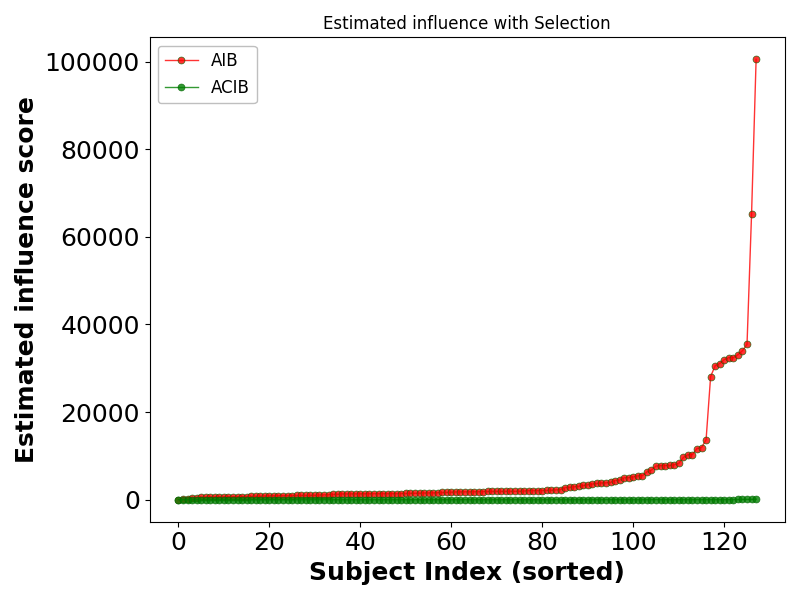}
    \end{minipage}%
    }%
    \caption{Estimated influence score in the proposed AIC and ACIB approaches.}
    \label{fig.4_add}
\end{figure*}

\begin{figure*}[htbp]
    \centering
    \subfloat[Linear]{
    \begin{minipage}[t]{0.33\linewidth}
    \centering
    \includegraphics[width=1.0\textwidth]{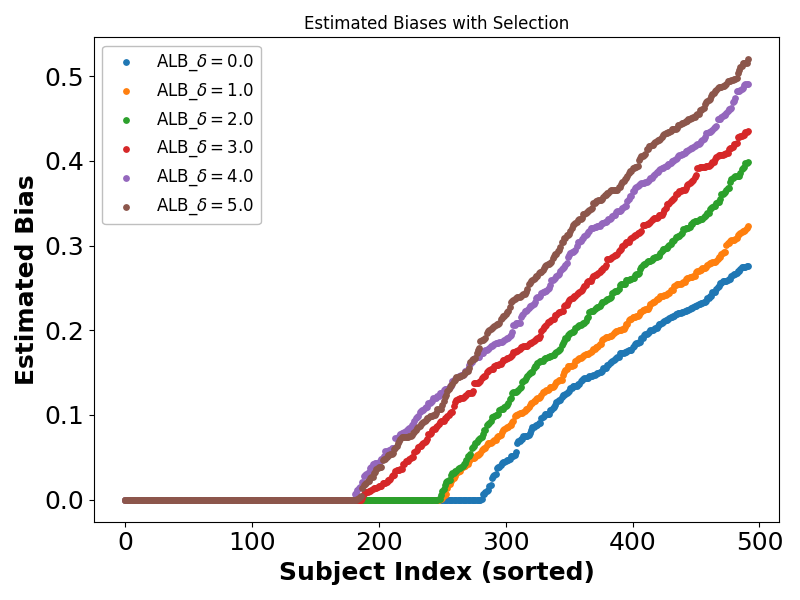}
    \end{minipage}%
    }%
    \subfloat[Nonlinear]{
    \begin{minipage}[t]{0.33\linewidth}
    \centering
    \includegraphics[width=1.0\textwidth]{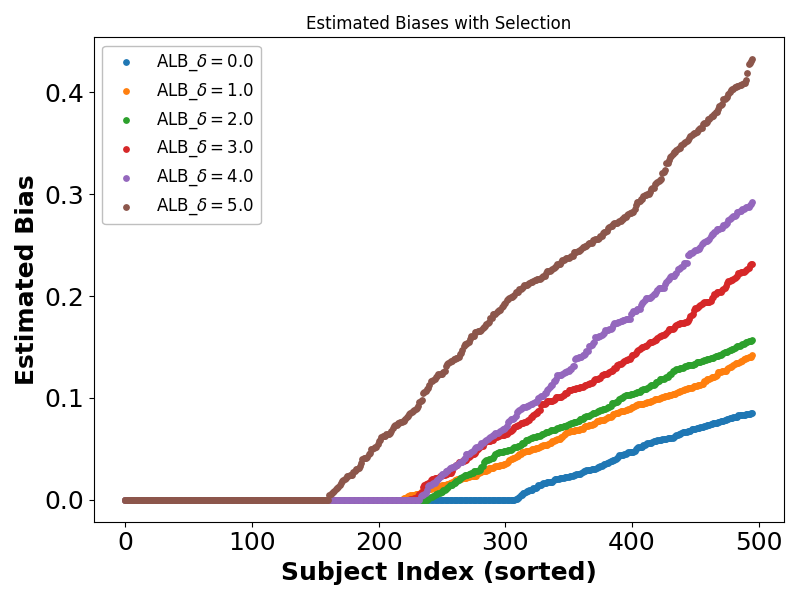}
    \end{minipage}%
    }%
    \subfloat[NSW $\&$ PSID]{
    \begin{minipage}[t]{0.33\linewidth}
    \centering
    \includegraphics[width=1.0\textwidth]{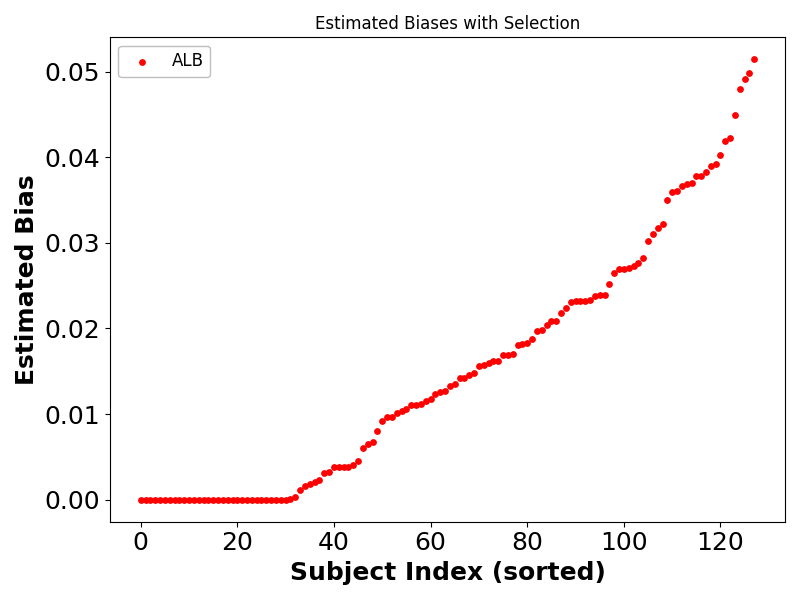}
    \end{minipage}%
    }%
    \caption{Estimated bias in the Lasso-based approach.}
    \label{fig.5_add}
\end{figure*}

\subsection{Some influence score estimation and bias estimation Results}

In this section, we present the influence score estimates $\{\mathcal{IF}(Z_j),j\in\mathcal{E}\}$ obtained from the proposed AIB and ACIB methods, as well as the bias estimates $\{\tilde b(X_j),j\in\mathcal{E}\}$ from theALB method. All results are reported after sorting. 

First, Figure~\ref{fig.4_add} shows the influence score estimation for the AIB and ACIB methods. In Figure~\ref{fig.4_add}(a) and (b), we display the differences in influence scores under different values of $\delta$ for the simulated datasets. The y axis corresponds to $$\{\bm{\mathcal {IF}} (Z) _ {\delta = a} - \bm{\mathcal {IF}} (Z) _ {\delta = 0} \},~a=\{0,1,2,3,4,5\},$$ where $\delta=0$ represents the setting without inconcurrency bias and thus the number of comparable samples is the largest. Consequently, this difference reflects the reduction in sample comparability caused by the inconcurrency bias. We observe that as $\delta$ increases, the influence score differences for the AIB method grow, indicating that the number of comparable samples decreases. In contrast, the influence score differences for the ACIB method remain essentially stable around zero across all $\delta$, suggesting that ACIB effectively removes the inconcurrency bias, and thereby improves sample utilization. In Figure~\ref{fig.4_add}(c), for the real dataset, we also observe that the estimated influence scores under AIB are larger than those under ACIB, further indicating that ACIB can mitigate inconcurrency bias and achieve outcome calibration.

Second, Figure~\ref{fig.5_add}(a) and (b) present the bias estimation of the ALB method under different values of $\delta$ for the simulated datasets. We observe that as $\delta$ increases, the number of comparable samples gradually decreases. This is expected, since a larger $\delta$ corresponds to stronger inconcurrency bias, resulting in a smaller number of comparable samples.

\subsection{Additional Results for Application}

{\bf Summary Statistics.} From Table~IV of the manuscript, we find that there are large differences in the distribution of student-level covariates between the control group in NSW and PSID, indicating the presence of covariate shift. Moreover, there is a notable disparity in the mean of the outcome variable: the sample mean of $Y$ in the NSW control group is $4.55$, whereas in the PSID data it is $5.28$. Specifically, the earning variable (with relatively large values) and the age and education variables (with comparatively smaller ranges) exhibit substantial differences in scale. To ensure comparability across features and prevent variables with larger magnitudes from dominating the analysis, we standardized these continuous variables. In contrast, categorical variables such as Race and Marital Status were not standardized, as they are already represented in binary (0/1) form, which inherently maintains consistent scaling across categories.

\end{document}